\pdfoutput=1
\documentclass[11pt]{article}

\usepackage{natbib}
\usepackage{geometry}
\usepackage[utf8]{inputenc}
\usepackage{times}
\usepackage{soul}
\usepackage{url}
\usepackage{hyperref}
\hypersetup{hidelinks}
\usepackage[utf8]{inputenc}
\usepackage[small]{caption}
\usepackage{graphicx}
\usepackage{amsmath}
\usepackage{amsthm}
\usepackage{booktabs}
\usepackage{algorithm}
\usepackage{algorithmic}
\usepackage{amsfonts}
\usepackage{xcolor}
\usepackage{multicol}
\usepackage{pgfplots}
\pgfplotsset{compat=newest}
\usepgfplotslibrary{colorbrewer}
\usepackage{pgfplotstable}
\usetikzlibrary{pgfplots.groupplots}

\usepackage{braket}
\usepackage{dsfont}
\usepackage{longtable}

\newtheorem{example}{Example}
\newtheorem{theorem}{Theorem}
\newtheorem{definition}{Definition}
\newtheorem{remark}{Remark}
\newtheorem{note}{Note}

\usepackage{subfig}
\usepackage{cleveref}
\usepackage{xcolor}

\newcommand{\wmc}{\ensuremath{\mathsf{WMC}}}
\newcommand{\wfomc}{\ensuremath{\mathsf{WFOMC}}}
\newcommand{\fotwo}{\ensuremath{\textbf{FO}^2}}
\newcommand{\ctwo}{\ensuremath{\textbf{C}^2}}
\newcommand{\barw}{\ensuremath{\overline{w}}}
\newcommand{\hb}{\ensuremath{\mathsf{HB}}}
\newcommand{\sm}{Smokes}
\newcommand{\fr}{Friends}
\newcommand{\edge}{E}
\newcommand{\heads}{Heads}

\newcommand{\Nat}{\mathbb{N}}

\newcommand{\Mat}[1]{\ensuremath{\textbf{#1}}}
\newcommand{\spectrum}[0]{\mathfrak{S}}

\title{On Discovering Interesting Combinatorial Integer Sequences}
\author{
Martin Svato\v{s}$^1$\and
Peter Jung$^1$\and
Jan T\'{o}th$^1$\and 
Yuyi Wang$^{2,3}$\and \\
Ond\v{r}ej Ku\v{z}elka$^1$\\ \\
$^1$Czech Technical University in Prague, Czech Republic\\
$^2$CRRC Zhuzhou Institute, China\\
$^3$ETH Zurich, Switzerland\\
\{svatoma1, jungpete, tothjan2, ondrej.kuzelka\}@fel.cvut.cz\\
yuyiwang920@gmail.com
}
\date{January 2023}

\begin{document}

\maketitle

\begin{abstract}
     We study the problem of generating {\em interesting} integer sequences with a combinatorial interpretation. For this we introduce a two-step approach. In the first step, we generate first-order logic sentences which define some combinatorial objects, e.g., undirected graphs, permutations, matchings etc. In the second step, we use algorithms for lifted first-order model counting to generate integer sequences that count the objects encoded by the first-order logic formulas generated in the first step. For instance, if the first-order sentence defines permutations then the generated integer sequence is the sequence of factorial numbers $n!$. We demonstrate that our approach is able to generate interesting new sequences by showing that a non-negligible fraction of the automatically generated sequences can actually be found in the Online Encyclopaedia of Integer Sequences (OEIS) while generating many other similar sequences which are not present in OEIS and which are potentially interesting. A key technical contribution of our work is the method for generation of first-order logic sentences which is able to drastically prune the space of sentences by discarding large fraction of sentences which would lead to redundant integer sequences.
\end{abstract}

\section{Introduction}\label{sec:introduction}

In this paper we are interested in integer sequences. As its name suggests, an {\em integer sequence} is a sequence of integers $a_0$, $a_1$, $a_2$, $\dots$, where $a_i \in \mathbb{Z}$ for all $i \in \mathbb{N}$. Integer sequences are fundamental mathematical objects that appear almost everywhere in mathematics, ranging from enumerative combinatorics, where they count objects with certain properties, to mathematical analysis, where they define functions by means of Taylor series, and in many other areas as well. There is even an encyclopedia of them, called Online Encyclopedia of Integer Sequences (OEIS),\footnote{\url{https://oeis.org}, for a popular account of the place of OEIS in mathematics, we also refer to the article in Quanta Magazine: \url{https://www.quantamagazine.org/neil-sloane-connoisseur-of-number-sequences-20150806/}.} whose offline predecessor was established in 1964 by Neil Sloane \citep{oeis}. It contains more than 359k integer sequences, as of January 2023. OEIS contains sequences that are of interest to professional or amateur mathematicians. 

A typical mode of use of the OEIS database is as follows. Say, you work on a combinatorial problem, counting undirected graphs on $n$ vertices that have certain property that you care about, e.g., having all vertex-degrees equal to 3. You manage to compute the numbers of these graphs for several small values of $n$ and you start wondering if someone did not study the same sequence of numbers. So you take the values you computed and insert them into the search box on the OEIS homepage and hit search. After that you receive all hits into OEIS and if you are lucky, one of them will tell you something interesting about your problem---maybe somebody has already solved it or at least computed more elements of the sequence.\footnote{For instance, for undirected graphs with all vertex degrees equal to $3$, one of the hits in OEIS would be sequence A002829: Number of trivalent (or cubic) labeled graphs with $2n$ nodes.}

How do sequences get into OEIS? Sequences that are deemed {\em interesting} are manually submitted to OEIS by users. Here, what is {\em interesting} is obviously subjective to a large extent. However, this is also a limitation of OEIS---the first person to study certain sequence will not get much help by looking it up in OEIS. Many quite natural sequences are not contained in OEIS. For instance, as observed by \citep{barvinek}, it contains sequences counting 2-regular graphs properly colored by 2 colors, but not 2-regular graphs properly colored by 3 colors. There are many similar examples of interesting sequences missing from OEIS, which might be potentially useful for some users. This is also the motivation for the work we present here in which we develop an automated method for discovering arguably {\em interesting} integer sequences.

We focus on combinatorial sequences, i.e., sequences which count objects of size $n$ that have some given property, which are the subject of interest of enumerative combinatorics \citep{stanley1986enumerative}. Examples of such combinatorial sequences include sequences counting: subsets of an $n$-element set, graphs on $n$ vertices, connected graphs on $n$ vertices, trees on $n$ vertices, permutations on $n$ elements without fixpoints etc. In particular, we focus on combinatorial sequences of structures {\em that can be described using a first-order logic sentence}. 

There are several advantages of working with combinatorial enumeration problems expressed in first-order logic. First, even though it may sometimes require some effort, the first-order logic sentences can be interpreted by the human users. For instance, the sentence $\forall x\; \neg R(x,x) \wedge \forall x \forall y\; R(x,y) \Rightarrow R(y,x)$ can be interpreted as encoding undirected graphs without loops. Second, despite the fact that counting the models of first-order logic sentences is generally intractable \citep{Lifting/FO3-intractable}, there are well-characterized non-trivial fragments of first-order logic for which counting is tractable \citep{Lifting/FO2-domain-liftable,Lifting/Skolemization,Lifting/C2-domain-liftable} with fast implementations available \citep{Lifting/FO2-domain-liftable,Lifting/FO2-cell-graphs}.\footnote{\url{https://github.com/UCLA-StarAI/Forclift}, \url{https://www.comp.nus.edu.sg/~tvanbr/software/fastwfomc.tar.gz}} This means that we are able to compute the respective combinatorial sequences fast.\footnote{Computational complexity of integer sequences that count combinatorial objects is an active research direction in enumerative combinatorics, see, e.g., \citep{pak2018complexity}.}

Our method has two stages. First, we generate first-order logic sentences from a tractable fragment. Second, we compute sequences for each of the generated sentences and filter out sentences which give rise to redundant sequences. It turns out that the first step is critical. As we demonstrate experimentally later in this paper, if we generated sentences naively, i.e., if we attempted to generate all sentences of length at most $k$ that differ syntactically, we would have to compute such huge numbers of sequences, most of them redundant, that we would never be able to get to the interesting ones. In the present paper, we therefore focus mostly on describing the sentence-generating component of our system.

The rest of this paper is structured as follows. \Cref{sec:preliminary} describes the preliminaries from first-order logic. \Cref{sec:sequenceDB} describes our approach to construct a database of sentences and the respective integer sequences, which is evaluated in \Cref{sec:experiments}.  The paper ends with related work in \Cref{sec:relatedWork} and conclusion in \Cref{sec:conclusion}.

\section{Preliminaries}\label{sec:preliminary}

We work with a function-free subset of first-order logic.
The language is defined by a finite set of \textit{constants} $\Delta$, a finite set of \textit{variables} $\mathcal{V}$ and a finite set of \textit{predicates} $\mathcal{P}$.
An \textit{atom} has the form $P(t_1, \ldots, t_k)$ where $P \in \mathcal{P}$ and $t_i \in \Delta$ $\cup$ $\mathcal{V}$.
A \textit{literal} is an atom or its negation.
A \textit{formula} is an atom and a literal.
More complex formulas may be formed from existing formulas by logical connectives, or by surrounding them with a universal ($\forall x$) or an existential ($\exists x$) quantifier where $x \in \mathcal{V}$.
A variable $x$ in a formula is called $free$ if the formula contains no quantification over $x$.
A formula is called a $sentence$ if it contains no free variables.
A formula is called \textit{ground} if it contains no variables.

As is customary in computer science, we adopt the \textit{Herbrand semantics} \citep{Logic/Herbrand} with a finite domain.
We use \hb{} to denote the \textit{Herbrand base}, i.e., the set all ground atoms.
We use $\omega$ to denote a \textit{possible world}, i.e., any subset of \hb.
Elements of a possible world are assumed to be true, all others are assumed to be false.
A possible world $\omega$ is a model of a sentence $\phi$ (denoted by $\omega \models \phi$) if the sentence is satisfied in $\omega$.

\subsection{Weighted First-Order Model Counting}
To compute the combinatorial integer sequences, we make use of the weighted first-order model counting (WFOMC) problem \citep{Lifting/WFOMC}.

\begin{definition}{(Weighted First-Order Model Counting)}
\label{def:wfomc}
Let $\phi$ be a sentence over some relational language $\mathcal{L}$.
Let \hb{} denote the Hebrand base of $\mathcal{L}$ over some domain of size $n\in\Nat$.
Let $\mathcal{P}$ be the set of the predicates of the language $\mathcal{L}$ and
let $\mathsf{pred}: \hb \mapsto \mathcal{P}$ map each atom to its corresponding predicate symbol.
Let $w: \mathcal{P} \mapsto \mathds{R}$ and $\overline{w}: \mathcal{P} \mapsto \mathds{R}$ be a pair of \textit{weightings} assigning a \textit{positive} and a \textit{negative} weight to each predicate in $\mathcal{L}$.
We define
\begin{align*} \wfomc(\phi, n, w, \overline{w}) = \sum_{\omega \subseteq \hb:\omega \models \phi} \prod_{l \in \omega} w(\mathsf{pred}(l)) \prod_{l \in \hb \setminus \omega} \overline{w}(\mathsf{pred}(l)).
\end{align*}
\end{definition}

\begin{example}
Consider the sentence $$\phi = \forall x\; \neg \edge(x,x)$$
and the weights $w(\edge) = \barw(\edge) = 1$.
Since all the weights are unitary, we simply count the number of models of $\phi$.
We can interpret the sentence as follows:
Each constant of the language is a vertex.
Each atom $\edge(A,B)\in\hb$ with $A,B\in\Delta$ denotes an edge from $A$ to $B$.
Furthermore, the sentence prohibits reflexive atoms, i.e, loops.
Overall, the models of $\phi$ will be all directed graphs without loops on $n$ vertices.
Hence, we obtain $$\wfomc(\phi, n, w, \barw) = 2^{n^2-n}.$$
\end{example}
\begin{example}
Consider the sentence $$\phi = \exists x\; \heads(x)$$
and the weights $w(\heads) = 4, \barw(\heads) = 1$.
Now, we can consider each domain element to be the result of a coin flip.
The sentence requires that there is at least one coin flip with the value of ``heads'' (there exists a constant $A\in\Delta$ such that $\heads(A)$ is an element of the model).
Suppose we have $i>0$ ``heads'' in the model.
Then, the model's weight will be $4^i\cdot 1^{n-i}=4^i$ and there will be $\binom{n}{i}$ such models.
Therefore, $$\wfomc(\phi, n, w, \barw) = \sum_{i=1}^{n}4^i\cdot\binom{n}{i} = 5^n - 1.$$
\end{example}

\subsection{WFOMC in the Two-Variable Fragment}
In order to make our calculations tractable, we limit the number of variables in each sentence to at most two.
Such language is known as \fotwo{} and it allows computing WFOMC in time polynomial in the domain size 
\citep{Lifting/FO2-domain-liftable,Lifting/Skolemization}.
We provide a brief overview of that tractability result.

When computing WFOMC in a lifted manner, we seek to avoid grounding the problem as much as possible.
Grounding first-order sentences often exponentially enlarges the problem and inherently leads to many symmetrical subproblems.
\begin{example}
\label{ex:grounding}
Consider the sentence $$\phi = \forall x\; (\sm(x) \Rightarrow Cancer(x)).$$
Grounding the sentence on the domain of size $n\in\Nat$ will produce a conjunction of $n$ implications.
Each of those implications will have three models with atoms completely different from the atoms in models of the other implications.
Moreover, there will be bijections between the models of different implications.
Overall, we could have computed the model count in a much simpler way.
For one particular constant, there will be three distinct models.
Since there are $n$ constants, the final model count will be $3^n$.
\end{example}

\noindent To avoid repeating the computations on such symmetrical instances, we aim to decompose the WFOMC problem into mutually independent parts with each needed to be solved only once.
{\em Cells} of a logical sentence whose WFOMC is to be computed allow such decomposition.
\begin{definition}[Cell]
A cell of an \fotwo{} formula $\phi$ is a maximal consistent 
conjunction of literals formed from atoms in $\phi$ using only
a single variable.
\end{definition}
\begin{example}
Consider the formula $$\phi = \sm(x) \wedge \fr(x, y) \Rightarrow \sm(y).$$ 
Then there are four cells:
\begin{align*}
C_1(x) &= \sm(x) \wedge \fr(x,x), \\
C_2(x) &= \neg \sm(x) \wedge \fr(x,x),\\
C_3(x) &= \neg \sm(x) \wedge \neg \fr(x,x),\\
C_4(x) &= \sm(x) \wedge \neg \fr(x,x).
\end{align*}
\end{example}

\noindent To simplify the WFOMC computation, we condition on cells in the following way:
\begin{align*}
    \psi_{ij}(x,y) &= \phi(x,y) \wedge \phi(y,x) \wedge C_i(x) \wedge C_j(y), \\
    \psi_{k}(x) &= \phi(x,x) \wedge C_k(x).
\end{align*}
And we compute
\begin{align*}
    r_{ij} &= \wmc(\psi_{ij}(A, B), w', \overline{w}'),\\
    w_{k} &= \wmc(\psi_{k}(A), w, \overline{w}),
\end{align*}
where WMC is simply the propositional version of WFOMC, $A,B\in\Delta$ and the weights $(w',\overline{w}')$ are the same as $(w,\overline{w})$ except for the atoms appearing in the cells conditioned on.
Those weights are set to one, since the weights of the unary and binary reflexive atoms are already accounted for in the $w_k$ terms.
Note that $r_{ij} = r_{ji}$.

Now, assuming there are $p$ distinct cells, we can write
\begin{align}
\label{eq:2wfomc}
\wfomc(\phi, n, w, \barw) = \sum_{\Mat{k}\in\Nat^p:|\Mat{k}|=n} \binom{n}{\Mat{k}}&\prod_{i,j\in[p]:i<j}r_{ij}^{(\Mat{k})_i(\Mat{k})_j}\prod_{i\in[p]}r_{ii}^{\binom{(\Mat{k})_i}{2}}w_i^{(\Mat{k})_i}.
\end{align}

However, the approach above is only applicable for universally quantified \fotwo{} sentences.
To get rid of existential quantifiers in the input formula, we can utilize specialized \textit{Skolemization} for WFOMC \citep{Lifting/Skolemization}.
The procedure eliminates existential quantifiers by introducing new (fresh) \textit{Skolem} predicates $Sk$ with $w(Sk)=1$ and $\barw(Sk)=-1$.

\begin{example}
\label{ex:skolem}
    Consider formulas 
    \begin{align*}
        \phi &= \forall x \exists y\; \edge(x, y),\\
        \psi &= \forall x \forall y\; \neg\edge(x, y) \vee Sk(x).
    \end{align*}
    It holds that 
    \begin{equation*}
     \wfomc(\phi, n, w, \barw) = \wfomc(\psi, n, w', \barw'),
    \end{equation*}
    where
    \begin{equation*}
        \begin{aligned}[c]
            w(\edge) &= w'(\edge),\\
            \barw(\edge) &= \barw'(\edge),
        \end{aligned}
        \qquad
        \begin{aligned}[c]
            w'(Sk)&=1,\\
            \barw'(Sk)&=-1.      
        \end{aligned}
    \end{equation*}
We refer the readers to {\em \citep{Lifting/Skolemization}} for justification.
\end{example}

\noindent Due to \Cref{eq:2wfomc} combined with the specialized Skolemization procedure, WFOMC can be evaluated in time polynomial in $n$ for any \fotwo{} sentence.

\subsection{WFOMC in the Two-Variable Fragment with Counting Quantifiers}
Although the language of \fotwo{} permits a polynomial-time WFOMC computation, its expressive power is naturally quite limited.
The search for a larger logical fragments still permitting a polynomial complexity is a subject of active research.
One possibility to extend the \fotwo{} while preserving its tractable property is by adding \textit{counting quantifiers}.
Such language is known as \ctwo{} and its tractability was shown by \cite{Lifting/C2-domain-liftable}.

Counting quantifiers are a generalization of the traditional existential quantifier.
For a variable $x \in \mathcal{V}$, we allow usage of a quantifier of the form $\exists^{=k}x$, where $k \in \Nat$.%
\footnote{\cite{Lifting/C2-domain-liftable} actually proved the tractability for a more general version of the counting quantifiers, i.e., $\exists^{\bowtie k}x$, where $\bowtie \in \Set{<,\leq, =, \geq,>}$. However, the counting with inequalities does not scale very well---in fact, even the equalities turn out to be computationally challenging---so we only work with equality in our counting quantifiers.}
Satisfaction of formulas with counting quantifiers is defined naturally.
For example, $\exists^{=k}x\; \psi(x)$ is satisfied in $\omega$ if there are exactly $k$ constants $\Set{A_1, A_2, \ldots, A_k} \subseteq \Delta$ such that $\omega \models \psi(A_i)$ if and only if $1\leq i\leq k$.

To handle counting quantifiers, \cite{Lifting/C2-domain-liftable} suggested evaluating WFOMC repeatedly on many {\em points} (domains). 
The values would be subsequently used in a polynomial interpolation.
Instead, we work with symbolic weights, which is, for all our purposes, equivalent to the polynomial interpolation.
We simply obtain the would-be-interpolated polynomial directly.%
\footnote{In \Cref{def:wfomc}, we defined WFOMC for real-valued weights only.
However, the extension to (multivariate) polynomials is natural and does not break anything.}

When we have a sentence with counting quantifiers whose WFOMC is to be computed, first thing to do is to convert the counting quantifiers to the traditional existential quantifier.
That can be achieved using yet another syntactic construct known as {\em cardinality constraints}.
We allow the formula to contain an {\em atomic formula} of the form $(|P| = k)$,%
\footnote{Similarly to counting quantifiers, cardinality constraints can be generalized to $(|P| \bowtie k)$ with $\bowtie\;\in\{<,\leq,=,\geq\,>\}$. See \citep{Lifting/C2-domain-liftable} for the full treatment.}
where $P\in\mathcal{P}$ is a predicate and $k\in\Nat$.
Intuitively speaking, a cardinality constraint enforces that all models of a sentence contain exactly $k$ atoms with the predicate~$P$.

\begin{example}
    \label{ex:ccs}
    Consider the sentences 
        \begin{align*}
            \phi &= \forall x \exists^{=1} y\; \edge(x, y),\\
            \phi' &= (\forall x \exists y\; \edge(x, y)) \land (|\edge| = n).
        \end{align*}
    Then it holds that
    $$\wfomc(\phi, n, w, \barw) = \wfomc(\phi', n, w, \barw)$$
    for any weights $(w, \barw)$.
\end{example}

\noindent Using transformations such as the one shown in \Cref{ex:ccs}, WFOMC of a \ctwo{} sentence $\phi$ can be reduced to WFOMC of the sentence 
$$\phi' = \psi \land \bigwedge_{i=1}^m (|P_i|=k_i),$$ 
where $\psi$ is an \fotwo{} sentence.
Then, for each cardinality constraint $(|P_i| = k_i)$, we define $w'(P_i) = x_i$, where $x_i$ is a new symbolic variable.
For predicates $Q\in\mathcal{P}$, which do not occur in any cardinality constraint, we leave the positive weight unchanged, i.e., $w'(Q)=w(Q)$.

Finally, we are ready to compute $\wfomc(\psi, n, w', \barw)$.
The result will be a multivariate polynomial over the symbolic variables introduced for each cardinality constraint.
However, only one of its monomials will carry the information about the actual WFOMC of the original \ctwo{} sentence.\footnote{When using counting quantifiers $\exists^{=k}$ with $k > 0$, we also need to take care of overcounting, which is described in detail in \citep{Lifting/C2-domain-liftable}.}
Namely, the monomial $$A\cdot\prod_{i=1}^m x_i^{e_i}$$ such that $e_i = k_i$ for each cardinality constraint $(|P_i| = k_i)$.

Now, we can report the final WFOMC result of the original \ctwo{} sentence $\phi$.
Nevertheless, we must still account for the positive weights that were replaced by symbolic variables when dealing with cardinality constraints.
Hence,
$$\wfomc(\phi, n, w, \barw) = A \cdot \prod_{i=1}^m w(P_i)^{k_i}.$$

\begin{example}\label{example:counting}
Consider the sentence $$\phi = \forall x \exists^{=1} y\; \edge(x, y),$$ domain of size $n=5$ and $w(\edge)=\barw(\edge)=1$.
Let us compute $\wfomc(\phi, n, w, \barw)$.

First, we get rid of the counting quantifier:
$$\phi' = (\forall x \exists y\; \edge(x, y)) \land (|\edge|=5)$$
Second, we introduce a symbolic variable $x$ as the positive weight of the $\edge$ predicate:
$$w'(\edge)=x$$
Finally, we evaluate $\wfomc(\forall x \exists y\; \edge(x, y), n, w', \barw)$ and extract the coefficient of the term where $x$ is raised to the fifth power:
$$\wfomc(\phi, n, w, \barw) = 3125.$$

Let us check the obtained result.
We can interpret the formula $\phi$ as a directed graph with each vertex having exactly one outgoing edge.
For each vertex, there are $n$ vertices that it could be connected to. Hence, there are $n^n$ such graphs.
For $n=5$, we obtain $5^5=3125$.
\end{example}

\section{Constructing the Sequence Database}\label{sec:sequenceDB}

Our aim is to build a database consisting of first-order logic sentences and the respective integer sequences that are generated by these sentences. We do not want the database to be exhaustive in terms of sentences. For any integer sequence, there may be many sentences that generate it\footnote{Trivially, if we allowed arbitrary predicate names, we could even say that there are infinitely many sentences that generate the same sequence.} and we want only one sentence per integer sequence. If there are multiple sentences that generate the same integer sequence, we call them {\em redundant}. We generally try to avoid generating redundant sentences. 

The database is constructed in two steps. In the first step, we generate first-order logic sentences and in the second step we compute the integer sequences that count the models of these sentences. In this section, we describe these two steps in the reverse order. First we describe how the sequences, which we will call {\em combinatorial spectra}, are computed from first-order logic sentences, which can be done using existing lifted inference methods. Then we describe our novel sentence-generation algorithm which strives to generate as few redundant sentences as possible.

\subsection{Computing the Integer Sequences}

Given a first-order logic sentence, we need to compute a number sequence such that its $k$-th member is the model count of a relational sentence on the domain of size $k$.
The set of domain sizes, for which the sequence member would be non-zero is called a \textit{spectrum} of the sentence.
{\em Spectrum} of a logical sentence $\phi$ is the set of natural numbers occurring as size of some finite model of $\phi$ \citep{borger2001classical}.
Since the sequence that we seek builds, in some sense, on top of the spectrum, and since the sequence can also be described as the result of the combinatorial interpretation of the original sentence, we dub the sequence \textit{combinatorial spectrum} of the sentence.

\begin{definition}[Combinatorial Spectrum]
    Combinatorial spectrum of a logical sentence $\phi$, denoted as $\spectrum(\phi)$, is a sequence of model counts of $\phi$ on finite domains of sizes taking on values $1,2,3,4,\dots$
\end{definition}

\begin{example}
    Consider again the sentence $$\phi = \exists x\; \heads(x).$$
    Then, all subsets of \hb{} are a model of $\phi$ except for the empty set.
    Hence, for a domain of size $n$, there will be $2^n-1$ models, i.e,
    \begin{equation*}
    \spectrum(\phi) = 1,3,7,15,\ldots
    \end{equation*}
    
\end{example}

Combinatorial spectra can be computed using a WFOMC algorithm. In this work we use our implementation of the algorithm from \citep{Lifting/FO2-cell-graphs}, which is a state-of-the-art algorithm running in time polynomial in $n$ for \fotwo{}. We use it together with our implementation of the reductions from \citep{Lifting/C2-domain-liftable} which allow us to compute spectra of any \ctwo{} sentence in time polynomial in $n$.

\subsection{Generating the First-Order Logic Sentences}\label{sec:sfinder}

In general, we aim to generate sentences that have the following syntactic form:\footnote{There are at least two Scott normal forms for \ctwo{} appearing in the literature \citep{gradel1999logics,pratt2009data}, which would allow us to use less quantifier prefixes. However, these normal forms were not designed for combinatorial counting---they were designed only to guarantee equisatisfiability of \ctwo{} sentences and their normal forms and they do not guarantee {\em combinatorial equivalence}. That is why they would not be directly useful for us in this paper.} 
\begin{align}
\bigwedge_{\begin{array}{c} Q_1 \in \{ \forall, \exists, \exists^{=1}, \dots, \exists^{=K} \},\\ Q_2 \in \{ \forall, \exists, \exists^{=1}, \dots, \exists^{=K} \} \end{array}} \bigwedge_{i=1}^{M} Q_1 x Q_2 y \; \Phi^{Q1,Q2}_i(x,y) \wedge \notag \\ 
\bigwedge_{Q \in \{ \forall, \exists, \exists^{=1}, \dots, \exists^{=K} \}} \bigwedge_{i=1}^{M'} Q x  \; \Phi^{Q}_i(x)\label{eq:sentenceForm}
\end{align}

\noindent where each $\Phi^{Q1,Q2}_i(x,y)$ is a quantifier-free disjunction of literals containing only the logical variables $x$ and $y$ and, similarly, each $\Phi^{Q}_i(x)$ is a quantifier-free disjunction of literals containing only the logical variable $x$. The integers $K$, $M$ and $M'$ are parameters.

Examples of sentences that have the form (\ref{eq:sentenceForm}) are:

\begin{itemize}
    \item $\forall x \exists^{=1} y\; R(x,y) \wedge \forall x \exists^{=1} y\; R(y,x)$, 
    \item $\forall x\;  \neg R(x,x) \wedge \forall x \forall y\; \neg R(x,y) \vee R(y,x)$.
\end{itemize}

\noindent Here the first sentence defines bijections (i.e., permutations) and the second sentence defines undirected graphs without loops. 

\begin{note}
We will slightly abuse terminology and use the term {\em clause} for the quantified disjunctions of the form $Q_1 x Q_2 y \; \Phi^{Q1,Q2}_i(x,y)$ and $Q x  \; \Phi^{Q}_i(x)$, even though the term {\em clause} is normally reserved only for universally quantified disjunctions.
\end{note}

\subsubsection{Do We Cover All of \ctwo{}?}

A natural question to ask is: Do we get all possible combinatorial spectra of \ctwo{} sentences if we restrict ourselves to sentences in the form of (\ref{eq:sentenceForm})? The answer seems to be negative, as we explain next, but it hardly matters in our opinion because the task that we set for ourselves in this paper is not to generate all combinatorial sequences of \ctwo{} sentences---this would not be feasible anyways because the number of different integer sequences generated as spectra of \ctwo{} sentences is infinite.\footnote{This is easy to see. In fact, even \fotwo{} sentences generate infinitely many integer sequences. Take, for instance the sentences $\varphi_k$ of the form $\varphi_k = \exists x\; \bigvee_{i=1}^k U_i(x)$. Their combinatorial spectra are $\spectrum{(\varphi_k)} = \left((2^k)^n-1\right)_{n=1}^\infty$. Hence, we have infinitely many combinatorial spectra even for these simple sentences---one for each $k \in \mathbb{N}$.} Instead, what we want to achieve is to generate as many {\em interesting} integer sequences as possible within a limited time budget.

Now we briefly explain why sentences of the form (\ref{eq:sentenceForm}) do not guarantee that we would be able to find all \ctwo{} combinatorial spectra. First of all, we cannot rely on normal forms from \citep{gradel1999logics,pratt2009data} because those were not designed to preserve model counts. While the transformation presented in \citep{Lifting/C2-domain-liftable} allows one to reduce the computation of model counts of any \ctwo{} sentence to a computation with sentences that are in the form of (\ref{eq:sentenceForm}), it requires some of the predicates to have negative weights. We do not allow negative weights in the generated sentences because they make the post-hoc combinatorial explanation of the sentences significantly more difficult.

\subsubsection{Traversing the Sentence Space}

We use a standard breadth-first search algorithm to traverse the space of \ctwo{} sentences. 
The algorithm starts with the empty sentence. In each layer of the search tree it generates all possible sentences that can be obtained by adding a literal to one of the sentences generated in the previous layer. The literal may be added into an existing clause or it can be added to the sentence as a new clause, in which case it also needs to be prefixed with quantifiers.

\begin{example}
Suppose we have the sentence $\varphi = \forall x \exists y \; R(x,y)$, which we want to extend. Suppose also that the only predicate in our language is the binary predicate $R$ and that the only allowed quantifiers are $\forall$ and $\exists$ (for simplicity). 
To extend $\varphi$, the first option we have is to add a new $R$-literal to the clause $\forall x \exists y \; R(x,y)$. There are 8 ways to do this resulting in the following sentences: $\varphi_1 = \forall x \exists y \; (R(x,y) \vee R(x,x))$, $\varphi_2 = \forall x \exists y \; (R(x,y) \vee R(x,y))$, $\varphi_3 = \forall x \exists y \; (R(x,y) \vee R(y,x))$, $\varphi_4 = \forall x \exists y \; (R(x,y) \vee R(y,y))$, $\varphi_5 = \forall x \exists y \; (R(x,y) \vee \neg R(x,x))$, $\varphi_6 = \forall x \exists y \; (R(x,y) \vee \neg R(x,y))$, $\varphi_7 = \forall x \exists y \; (R(x,y) \vee \neg R(y,x))$, $\varphi_8 = \forall x \exists y \; (R(x,y) \vee \neg R(y,y))$. 
The second option is to create a new single-literal clause and add it to $\varphi$. In this case we have the following: $\varphi_9 = \forall x \exists y \; R(x,y) \wedge \forall x\; R(x,x)$, $\varphi_{10} = \forall x \exists y \; R(x,y) \wedge \exists x\; R(x,x)$, $\varphi_{11} = \forall x \exists y \; R(x,y) \wedge \forall x\; \neg R(x,x)$, $\varphi_{12} = \forall x \exists y \; R(x,y) \wedge \forall x\; \neg R(x,x)$, and then sentences of one of the following types: $\forall x \exists y \; R(x,y) \wedge Q_1 x Q_2 y \; R(x,y)$, $\forall x \exists y \; R(x,y) \wedge Q_1 x Q_2 y \; R(y,x)$, $\forall x \exists y \; R(x,y) \wedge Q_1 x Q_2 y \; \neg R(x,y)$, and $\forall x \exists y \; R(x,y) \wedge Q_1 x Q_2 y \; \neg R(y,x)$ where $Q_1,Q_2 \in \{\forall,\exists\}$.
\end{example}

    \noindent As can be seen from this example, the branching factor is large even when the first-order language of the sentences contains just one binary predicate. However, if we actually computed the combinatorial spectra of these sentences, we would see that many of them are {\em redundant} (we give a precise definition of this term in the next subsection). Furthermore, if we were able to detect which sentences are {\em redundant} without computing their spectra, we would save a significant amount of time. This is because computation of combinatorial spectra, even though polynomial in $n$, is still computationally expensive. Moreover, if we were able to remove some sentences from the search, while guaranteeing that all {\em non-redundant} sentences would still be generated, we would save even more time. In the remainder of this section, we describe such techniques---either techniques that mark sentences as just {\em redundant}, in which case we will not compute their combinatorial spectra, or also as {\em safe-to-delete}, in which case we will not even use them to generate new sentences. We will use the term {\em not-safe-to-delete} when we want to refer to sentences which are {\em redundant} but not {\em safe-to-delete}.

\subsubsection{What Does It Mean That a Sequence Is Redundant?}

Given a collection of sentences $\mathcal{S}$, a sentence $\varphi \in \mathcal{S}$ is considered {\em redundant} if there is another sentence $\varphi' \in \mathcal{S}$ and $\spectrum(\varphi) = \spectrum(\varphi')$, i.e., if the other sentence generates the same integer sequence. Since checking whether two sentences have the same combinatorial spectrum is computationally hard (we give details in the Appendix), we will only search for sufficient conditions for when two sentences generate the same spectrum.

Apart from the above notion of redundancy, we also consider a sentence $\varphi \in \mathcal{S}$ redundant if there are two other sentences $\varphi'$, $\varphi''$ such that $\spectrum(\varphi) = \spectrum(\varphi') \cdot \spectrum(\varphi'')$, where the product $\cdot$ is taken element-wise. The rationale is that when this happens, the set of models of $\varphi$ likely corresponds to the elements of the Cartesian product of the models of $\varphi'$ and $\varphi''$ (or, at least, there is a bijection between them), which is not combinatorially very interesting.\footnote{After all, one can always create such sequences in a post-processing step and interpret them as elements of the respective Cartesian products if one so desires.} 

\subsubsection{Detecting Redundant Sentences}

Now that we explained what we mean by {\em redundant sentences}, we can move on to methods for detecting whether a sentence is redundant and if it is then whether it is also {\em safe-to-delete}. We stress upfront that the methods described in this section will not guarantee detecting all redundancies. On the other hand, these methods will be sound---they will not mark non-redundant sentences as redundant. Some of the techniques will mark a sentence as redundant but they will not give us a witness for the redundancy, i.e., other sentences with the same combinatorial spectrum. This will be the case for techniques that guarantee that the witness is a shorter sentence (in the number of literals), which must have been generated earlier, thus, we will know that by pruning the longer redundant sentences, we will not affect completeness of the search.

\begin{table*}[p]
\caption{Pruning techniques used in the algorithm for generation of non-redundant sentences.}\label{tab:pruning}
\begin{tabular}{lp{10.7cm}
}
   \bf Pruning Technique Name  &  \bf Short Description of the Idea 
   \\ \hline \hline
   \multicolumn{2}{l}{
   \begin{minipage}{0.9\linewidth}
   \vspace{0.1cm}
   {\em Sentences detected using the techniques described below are \emph{safe-to-delete}.}
   \end{minipage}}
   \\ \hline
    \em Isomorphic Sentences & Two sentences are isomorphic if one can be obtained from the other by renaming variables and predicate names. {\em See Appendix for details}. 
    \\ \hline
    \em Decomposable Sentences & If a sentence $\varphi$ can be written as a conjunction $\varphi = \varphi' \wedge \varphi''$ of two conjunctions with disjoint sets of predicates then $\varphi$ is redundant. {\em See the main text for a justification}. 
    \\ \hline
    \em Tautologies \& Contradictions & Any sentence which contains an always-true (i.e., tautological) clause is redundant---there exists a shorter sentence with the same combinatorial spectrum. All unsatisfiable sentences (contradictions) produce the same combinatorial spectrum consisting of zeros and, hence, are also redundant. 
    \\ \hline
    \em Negations & If two sentences can be made isomorphic by negating all occurrences of literals of some predicates, then they generate the same combinatorial spectra. For example, $\spectrum(\forall x\exists y\; R(x,y)) = \spectrum(\forall x\exists y\; \neg R(x,y))$.
    \\ \hline
    \em Permuting Arguments & {\em Argument-flip} on a predicate $R$ is a transformation which replaces all occurrences of $R(x,y)$ by $R(y,x)$ and all occurrences of $R(y,x)$ by $R(x,y)$. If two sentences $\varphi$ and $\varphi'$ can be made isomorphic using argument flips, then they generate the same combinatorial spectra, i.e., $\spectrum(\varphi) = \spectrum(\varphi')$.
    \\ \hline \hline
    \multicolumn{2}{l}{\begin{minipage}{0.9\linewidth}
    \vspace{0.1cm}
    {\em Sentences detected by techniques described below are \emph{not-safe-to-delete}. We do not compute combinatorial spectra for those sentences and do not store them in the database.}
    \end{minipage}}
    \\ \hline
    \em Trivial Constraints & Suppose a sentence $\varphi$ contains a clause of the form $\forall x \; U(x)$ or $\forall x \forall y \; R(x,y)$, which we call a {\em trivial constraint}. Then the sentence $\varphi'$ obtained from $\varphi$ by dropping the trivial constraint and replacing all occurrences of $U$ or $R$, respectively by {\em true}, has the same combinatorial spectrum as $\varphi$. 
    \\ \hline
    \em Reflexive Atoms & If a binary literal $R$ appears in a sentence $\varphi$ only as $R(x,x)$, $\neg R(x,x)$, $R(y,y)$ or $\neg R(y,y)$ and $\varphi$ has at least two literals, then the sentence is redundant. {\em See the main text for a justification}. 
    \\ \hline
    \em Subsumption & If a sentence $\varphi$ contains two clauses $Q_1 x Q_2 y\; \alpha(x,y)$ and $Q_1 x Q_2 y\; \beta(x,y)$, with the same quantifier prefix, and if there is a substitution $\theta : \{x,y \} \rightarrow \{x,y\}$ such that $\alpha \theta \subseteq \beta$, then $\varphi$ is redundant---the sentence $\varphi'$ obtained from $\varphi$ by dropping $Q_1 x Q_2 y\; \beta(x,y)$ generates the same combinatorial spectrum. 
    \\ \hline
    \em Cell Graph Isomorphism & {\em See the main text}. 
    \\ \hline
\end{tabular}
\end{table*}

The pruning methods that are used by our algorithm for generation of sentences are summarized in Table \ref{tab:pruning}. Some of them are rather straightforward and do not require much further justification here. 

The method called {\em Isomorphic Sentences} is a straightforward extension of the methods for enumeration of non-isomorphic patterns, known from data mining literature (see, e.g., \citep{farmr}), where the main difference is that when checking isomorphism, we allow renaming of predicates (the details are described in the Appendix for completeness). 

The method called {\em Decomposable Sentences} is based on the following observation, which is well-known among others in lifted inference literature \citep{Lifting/FO2-domain-liftable}: Let $\varphi = \varphi_1 \wedge \varphi_2$ be a first-order logic sentence. If $\varphi_1$ and $\varphi_2$ use disjoint sets of predicates then it is not hard to show that $\spectrum(\varphi) = \spectrum(\varphi_1) \cdot \spectrum(\varphi_2)$, where the product is taken element-wise and $\spectrum(\varphi_1)$ and $\spectrum(\varphi_2)$ are understood to be computed only over the languages consisting of the predicates contained in $\varphi_1$ and $\varphi_2$, respectively. 

While the method called {\em Permuting Arguments} may not need a more detailed explanation per se, we will still illustrate it here on an example to provide a better intuition. Suppose that we have two sentences: $\varphi_1 = \forall x \exists y\; E(x,y)$ and $\varphi_2 = \forall x \exists y \; E(y,x)$. The first one can be interpreted as modelling directed graphs in which no vertex has out-degree $0$ and the second one as modelling directed graphs in which no vertex has in-degree 0. This interpretation was based on our decision to interpret $E(x,y)$ as an edge form $x$ to $y$, yet we could have also interpreted it as an edge from $y$ to $x$ and this would change nothing about the combinatorial spectrum of the sentence (which does not depend on how we interpret the sentence). If we generalize this observation, we realize that sentences that differ only in the order of arguments of some predicates (like $\varphi_1$ and $\varphi_2$ above) must generate the same combinatorial spectrum.

Next, we give little more detail on the method called {\em Reflexive Atoms}. If a sentence $\varphi$ contains atoms of some binary predicate $R$ only in the form $R(x,x)$ or $R(y,y)$ then all the ground atoms $R(i,j)$, where $i$ and $j$ are domain elements and $i \neq j$, are unconstrained by $\varphi$. It follows that $\spectrum(\varphi) = \spectrum(\varphi') \cdot \spectrum(\varphi'')$ where $\varphi' = \forall x\; \neg R(x,x)$ and $\varphi''$ is a sentence obtained by replacing all occurrences of $R(x,x)$ by $U_R(x)$ and occurrences of $R(y,y)$ by $U_R(y)$ where $U_R$ is a fresh predicate. Here, $\varphi'$ accounts for all possible configurations of the atoms $R(i,j)$ with arguments $i \neq j$. It follows that such a sentence $\varphi$ is redundant.

The methods {\em Tautologies \& Contradictions}, {\em Negations}, {\em Trivial Constraints} and {\em Subsumption}, do not need any further explanation beyond what is in Table \ref{tab:pruning}.
The only remaining method, namely {\em Cell Graph Isomorphism}, is described in the next subsection.

\subsubsection{Cell Graph Isomorphism}
The final pruning method of \Cref{tab:pruning} called {\em Cell Graph Isomorphism} relies on a concept from the area of lifted inference, originally intended for a more efficient evaluation of \Cref{eq:2wfomc}.
In \citep{Lifting/FO2-cell-graphs}, the authors introduced a special structure called a {\em cell graph} to help them compute WFOMC faster.

\begin{definition}[Cell Graph]
    A cell graph $G_\phi$ of a sentence $\phi$ is a complete graph $(V, E)$ such that
    \begin{enumerate}
        \item $V$ is the set of cell labels $\Set{1,2,\ldots,p}$,
        \item each node $i \in V$ has a label $w_i$,
        \item each edge, including loops, from node $i$ to $j$ has a label $r_{ij}$.
    \end{enumerate}
\end{definition}

\noindent As one can observe from \Cref{eq:2wfomc}, the WFOMC computation is fully determined by the terms $r_{ij}$ and $w_{k}$.
That remains unchanged even with the counting quantifiers, since then, only the symbolic result of \Cref{eq:2wfomc} is further searched for particular monomials.
Hence, the computation is fully determined by a cell graph, which contains all the $r_{ij}$ and $w_{k}$ values.

Building on that observation, we propose a pruning technique based on two cell graphs being isomorphic.
If cell graphs of two sentences are isomorphic, then their WFOMC results will be the same, and consequently, their combinatorial spectra will be the same as well.
We formalize those claims below.

We start by discussing the simpler case where all weights are real-valued.
That is enough to apply this pruning method to \fotwo{} sentences, and then we extend it to the case with symbolic weights, which is needed for correct handling of sentences from \ctwo{}.

First we define what we mean by cell graph isomorphism.

\begin{definition}[Cell Graph Isomorphism, for graphs with real-valued weights]  
    Let $G$ and $G'$ be two cell graphs where each edge $\{i,j\} \in E(G)$ ($\{i',j'\} \in E(G')$, respectively) is labeled by a real-valued weight $r_{ij}$ ($r_{i'j'}'$, respectively) and each vertex $i \in V(G)$ ($i' \in V(G')$) is labeled by a real number $w_i$ ($w_{i'}'$). We say that $G$ and $G'$ are isomorphic if there exists a bijection $f : V(G) \rightarrow V(G')$ such that $w_i' = w_{f(i)}$ and $r_{ij}' = r_{f(i),f(j)}$ for all $i,j \in V(G)$.
\end{definition}

\noindent In order to exploit the isomorphism, we will exploit the following interesting property of C\ref{eq:2wfomc}.

\begin{remark}\label{remark:key}
Let $p$ be again the number of cells in \Cref{eq:2wfomc} and let $$f : [p] \rightarrow [p]$$ be a bijection. Then the following equality holds:
\begin{align*}
    \sum_{\Mat{k}\in\Nat^p:|\Mat{k}|=n} \binom{n}{\Mat{k}}&\prod_{i,j\in[p]:i<j}r_{ij}^{(\Mat{k})_i(\Mat{k})_j}\prod_{i\in[p]}r_{ii}^{\binom{(\Mat{k})_i}{2}}w_i^{(\Mat{k})_i} = \\
    \sum_{\Mat{k}\in\Nat^p:|\Mat{k}|=n} \binom{n}{\Mat{k}}&\prod_{i,j\in[p]:i<j}r_{f(i),f(j)}^{(\Mat{k})_i(\Mat{k})_j}\prod_{i\in[p]}r_{f(i),f(i)}^{\binom{(\Mat{k})_i}{2}}w_{f(i)}^{(\Mat{k})_i}.
\end{align*}
In other words, permuting the cells, while preserving the structure of the weights, does not change the resulting value.
\end{remark}

Next we state the result which will justify using the cell graph isomorphism method for \fotwo{} sentences.

\begin{theorem}\label{thm:first}
Let $\phi$ and $\psi$ be two \fotwo{} sentences with weights $(w, \barw)$ and $(w',\barw')$, respectively, and let $G_{\phi}$ and $G_{\psi}$ be their respective cell graphs. If $G_{\phi}$ is isomophic to $G_{\psi}$, then 
$$\wfomc(\phi,n,w,\barw)=\wfomc(\psi,n,w,\barw)$$
for any domain size $n\in\Nat$.
\end{theorem}
\begin{proof}[Proof sketch]
The proof follows from the following observation: Let $f$ be the bijection $f : V(G_\phi) \rightarrow V(G_\psi)$ preserving weights, which must exist from the definition of cell graph isomorphism. Consider the equation for computing WFOMC of a sentence from its cell graph, \Cref{eq:2wfomc}.
If we apply the bijection $f$ on the cell indices ($i$'s and $j$'s from the equation), it will turn the equation for computing WFOMC of $\phi$ to the one for $\psi$ (again because $f$ is weight-preserving bijection). It follows from Remark \ref{remark:key} that these two must be the same and therefore WFOMC of $\phi$ and $\psi$ must be equal for any domain size~$n$.
\end{proof}

Next we extend the cell graph isomorphism method to \ctwo{} sentences. For that, we first need to extend the definition of cell graphs and of cell graph isomorphism.

\begin{definition}[Cell Graph with Cardinality Constraints]
    A cell graph $G_\phi$ of a sentence $\phi$ is a pair $(G,C)$ consisting of:
    \begin{enumerate}
        \item A complete graph $G = (V, E)$ such that
            \begin{enumerate}
                \item $V$ is the set of cell labels $\Set{1,2,\ldots,p}$,
                \item each node $i \in V$ has a label $w_i$,
                \item each edge, including loops, from node $i$ to $j$ has a label $r_{ij}$.
            \end{enumerate}
            Here the weights $w_i$ and $r_{ij}$ are, in general, multivariate polynomials.
        \item A set $C$ of monomials representing the cardinality constraints.
    \end{enumerate}
\end{definition}

\begin{example}
Consider the sentence 
$$\phi = (\forall x\; \neg \edge(x,x)) \land (\forall x \forall y\; \neg \edge(x, y) \vee \edge(y,x)) \land (|\edge|=10)$$ 
which models undirected graphs with 5 edges. There is only one cell which is consistent with $\phi$ for this sentence, $\neg \edge(x,x)$. As we already saw in \Cref{example:counting}, to encode the cardinality constraint $|\edge|=5$, we need to introduce the symbolic weight $w(\edge) = x$. The cell graph then consists of the graph given by the the weights $w_1 = 1$, $r_{1,1} = 1+x^2$, and of the singleton set $C = \{ x^{10} \}$, representing the cardinality constraint. 
\end{example}

Now we are ready to state the definition of cell graph isomorphism for \fotwo{} sentences with cardinality constraints (which is all we need to encode \ctwo{} sentences).

\begin{definition}[Cell Graph Isomorphism, for graphs with symbolic weights and cardinality constraints]    
Let $(G,C)$ and $(G',C')$ be two cell graphs with cardinality constraints where each edge $\{i,j\} \in E(G)$ ($\{i',j'\} \in E(G')$, respectively) is labeled by a multivariate polynomial $r_{ij}$ ($r_{i'j'}'$, respectively) and each vertex $i \in V(G)$ ($i' \in V(G')$) is labeled by a multivariate polynomial $w_i$ ($w_{i'}'$). We say that $G$ and $G'$ are isomorphic if there exists a bijection $f : V(G) \rightarrow V(G')$ and another bijection $g$ mapping variables occurring in the polynomials in $(G,C)$ to variables occurring in the polynomials in $(G',C')$ which satisfy the following conditions: 
\begin{enumerate}
    \item $w_i' = g(w_{f(i)})$,
    \item $r_{ij}' = g(r_{f(i),f(j)})$ for all $i,j \in V(G)$,
    \item $C' = g(C)$.
\end{enumerate}
\end{definition}

\noindent The above definition is more complicated than the one for cell graphs of \fotwo{} sentences because we need to make sure that when we discover an isomorphism of the cell graph, it will not ``break'' the cardinality constraints.

Finally we are ready to formally show that cell graph isomorphism can be used also for \ctwo{} sentences.

\begin{theorem}
Let $\phi$ and $\psi$ be two \ctwo{} sentences and $\phi'$ and $\psi'$ be their encoding into \fotwo{} sentences with cardinality constraints. Let $(G_{\phi'},C_{\phi'})$ and $(G_{\psi'},C_{\psi'})$ be their respective cell graphs with constraints. If $(G_{\phi'},C_{\phi'})$ is isomorphic to $(G_{\psi'},C_{\psi'})$, then 
$$\wfomc(\phi,n,w,\barw)=\wfomc(\psi,n,w,\barw)$$
for any domain size $n\in\Nat$ and any weights $(w, \barw)$.
\end{theorem}
\begin{proof}
The proof is a straightforward extension of the proof of Theorem \ref{thm:first}.
\end{proof}

\noindent Therefore, it is enough to output just one sentence from each equivalence class induced by cell graph isomorphism.
However, as is already stated in \Cref{tab:pruning}, sentences with isomorphic cell graphs are {\em not-safe-to-delete}, meaning that combinatorial spectrum is computed only for one member of the induced equivalence class, but all the members are used to further expand the search space.

\section{Experiments}\label{sec:experiments}

In this section, we experimentally evaluate the effectiveness of the techniques for constructing the database of integer sequences described in~\Cref{sec:sfinder} within a reasonable amount of time. 
Furthermore, we take a closer look at a few interesting generated \ctwo{} sentences whose combinatorial spectra appear in OEIS. Finally, we investigate a few sentences whose combinatorial spectra do not appear in OEIS.

\subsection{Filling the Database of Integer Sequences}\label{sec:dbfilling}

We ran two separate experiments with generators of \fotwo{} and \ctwo{} sentences. We set a time-limit of five minutes for the computation of combinatorial spectra per sentence; results for these experiments are depicted in \Cref{fig:fo52} and \Cref{fig:c52} for \fotwo{} and \ctwo{}, respectively. Our aim with these experiments was to assess the effect of the pruning techniques that we proposed. 

We started with a baseline consisting of just the method that filters out sentences which are isomorphic (using the standard notion of isomorphism used in pattern mining literature, which does not consider renaming predicates \citep{farmr}) and with pruning of \emph{decomposable sentences}---these are the very essentials any reasonable method would probably implement. Then we enhanced the baseline with \emph{Tautologies \& Contradiction}. In a similar fashion, we added a single pruning technique on top of the previous one in the following order: \emph{Isomorphic Sentences}, \emph{Negations}, \emph{Permuting Arguments}, \emph{Reflexive Atoms}, \emph{Subsumption}, \emph{Trivial Constraints}, and \emph{Cell Graph Isomorphism}.  It can be seen that our methods reduce both the runtime and the number of generated sentences by orders of magnitude.

The pruning techniques help to scale up the process of filling the database in two ways. Whereas the naive approach (e.g. \emph{baseline}) generates a lot of sentences fast, soon consuming all available memory, {\em safe-to-delete} techniques lower the memory requirements significantly. All pruning techniques consume some computation time, but that is negligible compared to the time needed for computing combinatorial spectra, which is the most time-demanding part of the task; see \Cref{fig:fo52}c  and \Cref{fig:fo52}b, respectively. Since the pruning methods, including those which are {\em not-safe-to-delete}, reduce the number of computations of combinatorial spectra, their use quickly pays off, as can be clearly seen from \Cref{fig:fo52}b and \Cref{fig:c52}b which show the estimated\footnote{Since the methods which do not use the full set of our pruning techniques, generate an extremely high number of (mostly redundant) sentences, computing their spectra would take thousands of hours. Therefore, we only estimated the runtime by computing the spectra only for a random sample of sentences for these methods.} time to fill in the database.

\Cref{fig:fo52}a and \Cref{fig:c52}a also show a lower bound on the number of unique combinatorial spectra,\footnote{Computing spectra for sentences with more than five literals is time demanding, so we estimated the upper levels with only a random sample of longer sentences. Therefore we show the lower bound of spectra only for the fully evaluated levels. } i.e. the minimum number of non-redundant sentences that would fill the database with the same number of unique integer sequences.

\pgfplotsset{scaled y ticks=false}%
\definecolor{patriarch}{rgb}{0.5, 0.0, 0.5}
\definecolor{shockingpink}{rgb}{0.99, 0.06, 0.75}
\definecolor{ao(english)}{rgb}{0.0, 0.5, 0.0}
\definecolor{palesilver}{rgb}{0.79, 0.75, 0.73}
\definecolor{navyblue}{rgb}{0.0, 0.0, 0.5}
\definecolor{mint}{rgb}{0.24, 0.71, 0.54}
\definecolor{saddlebrown}{rgb}{0.55, 0.27, 0.07}

\pgfplotscreateplotcyclelist{shortList}{
{black,thick},
{mint},
{orange},
{patriarch},
{shockingpink},
{saddlebrown},
{palesilver},
{blue},
{ao(english),thick},
{red,thick},
}

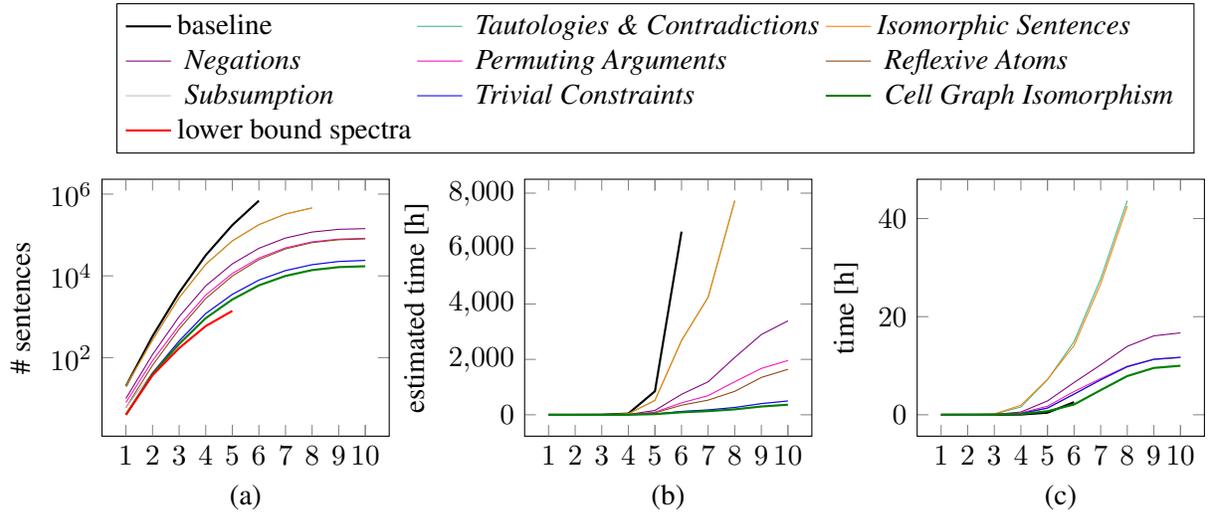
\begin{figure*}

\begin{tikzpicture}\begin{groupplot}[group style={group size=3 by 1},height=5cm,width=5.4cm,xlabel=layer,cycle list/Dark2,  xticklabels={},  extra x ticks={0,1,2,3,4,5,6,7,8,9,10} ]

	\nextgroupplot[ylabel={\# sentences},ymode=log, xlabel={(a)},cycle list name=shortList,legend style = {at = {(0.05, 1.1)}, anchor = south west, legend cell  align = left, legend columns = 3}]

		\addplot+[mark=none]	 coordinates {(1, 20.000000) (2, 334.000000) (3, 3890.000000) (4, 31820.000000) (5, 174188.000000) (6, 690522.000000)}; \addlegendentry{baseline} 
		\addplot+[mark=none]	 coordinates {(1, 20.000000) (2, 268.000000) (3, 2934.000000) (4, 19130.000000) (5, 71668.000000) (6, 178782.000000) (7, 325162.000000) (8, 459312.000000)
  }; \addlegendentry{\emph{ Tautologies \& Contradictions }} 
		\addplot+[mark=none]	 coordinates {(1, 20.000000) (2, 268.000000) (3, 2934.000000) (4, 19130.000000) (5, 71668.000000) (6, 178782.000000) (7, 325162.000000) (8, 459312.000000)
  }; \addlegendentry{\emph{Isomorphic Sentences}} 
		\addplot+[mark=none]	 coordinates {(1, 10.000000) (2, 119.000000) (3, 1002.000000) (4, 5688.000000) (5, 19634.000000) (6, 47236.000000) (7, 84165.000000) (8, 117932.000000) (9, 137330.000000) (10, 143054.000000)}; \addlegendentry{\emph{ Negations }} 
		\addplot+[mark=none]	 coordinates {(1, 8.000000) (2, 84.000000) (3, 616.000000) (4, 3386.000000) (5, 11437.000000) (6, 27314.000000) (7, 48444.000000) (8, 67939.000000) (9, 79285.000000) (10, 82829.000000)}; \addlegendentry{\emph{ Permuting Arguments }} 
		\addplot+[mark=none]	 coordinates {(1, 6.000000) (2, 65.000000) (3, 499.000000) (4, 2785.000000) (5, 9924.000000) (6, 24920.000000) (7, 45668.000000) (8, 65082.000000) (9, 76428.000000) (10, 79972.000000)}; \addlegendentry{\emph{ Reflexive Atoms }} 
		\addplot+[mark=none]	 coordinates {(1, 6.000000) (2, 44.000000) (3, 269.000000) (4, 1240.000000) (5, 3612.000000) (6, 7997.000000) (7, 13575.000000) (8, 18897.000000) (9, 22607.000000) (10, 23994.000000)}; \addlegendentry{\emph{ Subsumption }} 
		\addplot+[mark=none]	 coordinates {(1, 4.000000) (2, 42.000000) (3, 253.000000) (4, 1191.000000) (5, 3512.000000) (6, 7858.000000) (7, 13436.000000) (8, 18758.000000) (9, 22468.000000) (10, 23855.000000)}; \addlegendentry{\emph{ Trivial Constraints }} 
		\addplot+[mark=none]	 coordinates {(1, 4.000000) (2, 40.000000) (3, 216.000000) (4, 923.000000) (5, 2642.000000) (6, 5861.000000) (7, 9953.000000) (8, 13839.000000) (9, 16341.000000) (10, 17179.000000)}; \addlegendentry{\emph{ Cell Graph Isomorphism } } 
		\addplot+[mark=none]	 coordinates {(1, 4) (2, 37) (3, 171) (4, 590) (5, 1390)
  }; \addlegendentry{lower bound spectra} 
	\nextgroupplot[ylabel={estimated time [h]}, xshift=0.8cm, xlabel={(b)},cycle list name=shortList]

		\addplot+[mark=none]	 coordinates {(1, 0.000010) (2, 0.310768) (3, 3.440501) (4, 43.076306) (5, 849.055825) (6, 6611.269539)}; 
		\addplot+[mark=none]	 coordinates {(1, 0.000278) (2, 0.320202) (3, 1.877380) (4, 34.321044) (5, 519.864971) (6, 2682.359187) (7, 4242.635707) (8, 7728.464859)
  }; 
		\addplot+[mark=none]	 coordinates {(1, 0.000278) (2, 0.325202) (3, 1.910992) (4, 34.682155) (5, 519.950249) (6, 2681.561132) (7, 4241.635985) (8, 7727.358193)
  }; 
		\addplot+[mark=none]	 coordinates {(1, 0.000278) (2, 0.123255) (3, 0.865450) (4, 11.583262) (5, 156.844816) (6, 739.367980) (7, 1190.143820) (8, 2069.306610) (9, 2896.170438) (10, 3390.534099)}; 
		\addplot+[mark=none]	 coordinates {(1, 0.000278) (2, 0.115711) (3, 0.738394) (4, 7.462906) (5, 93.365818) (6, 430.481398) (7, 690.330469) (8, 1193.728745) (9, 1671.299822) (10, 1960.496256)}; 
		\addplot+[mark=none]	 coordinates {(1, 0.000010) (2, 0.108288) (3, 0.612594) (4, 5.535379) (5, 70.229089) (6, 345.422054) (7, 529.016193) (8, 845.754058) (9, 1347.605466) (10, 1640.470500)}; 
		\addplot+[mark=none]	 coordinates {(1, 0.000010) (2, 0.108288) (3, 0.377590) (4, 3.061012) (5, 31.214229) (6, 124.547823) (7, 180.303038) (8, 267.420315) (9, 406.085720) (10, 502.118766)}; 
		\addplot+[mark=none]	 coordinates {(1, 0.000278) (2, 0.093058) (3, 0.363471) (4, 2.893784) (5, 30.128430) (6, 121.518847) (7, 176.809969) (8, 263.927524) (9, 402.594318) (10, 498.626809)}; 
		\addplot+[mark=none]	 coordinates {(1, 0.020000) (2, 0.132224) (3, 0.387825) (4, 2.457154) (5, 23.130763) (6, 90.083454) (7, 131.803795) (8, 196.610278) (9, 298.420700) (10, 363.364821)}; 
	\nextgroupplot[ylabel={time [h]}, xshift=0.4cm, xlabel={(c)},cycle list name=shortList]

		\addplot+[mark=none]	 coordinates {(1, 0.000010) (2, 0.000288) (3, 0.004454) (4, 0.050843) (5, 0.429732) (6, 2.574732)}; 
		\addplot+[mark=none]	 coordinates {(1, 0.000278) (2, 0.009722) (3, 0.143889) (4, 1.589722) (5, 7.146111) (6, 14.921944) (7, 27.712222) (8, 43.644167)
  }; 
		\addplot+[mark=none]	 coordinates {(1, 0.000278) (2, 0.014722) (3, 0.177500) (4, 1.950833) (5, 7.231389) (6, 14.123889) (7, 26.712500) (8, 42.537500)
  }; 
		\addplot+[mark=none]	 coordinates {(1, 0.000278) (2, 0.006667) (3, 0.072222) (4, 0.561389) (5, 2.818889) (6, 6.584167) (7, 10.180556) (8, 13.951667) (9, 16.091389) (10, 16.691111)}; 
		\addplot+[mark=none]	 coordinates {(1, 0.000278) (2, 0.005556) (3, 0.048889) (4, 0.380278) (5, 1.750278) (6, 4.750000) (7, 7.376111) (8, 9.893611) (9, 11.319444) (10, 11.710556)}; 
		\addplot+[mark=none]	 coordinates {(1, 0.000010) (2, 0.001121) (3, 0.030288) (4, 0.261399) (5, 1.365566) (6, 4.201954) (7, 7.103343) (8, 9.767232) (9, 11.262232) (10, 11.661121)}; 
		\addplot+[mark=none]	 coordinates {(1, 0.000010) (2, 0.001121) (3, 0.030288) (4, 0.261954) (5, 1.370566) (6, 4.211121) (7, 7.129732) (8, 9.809732) (9, 11.311399) (10, 11.711677)}; 
		\addplot+[mark=none]	 coordinates {(1, 0.000278) (2, 0.003056) (3, 0.033333) (4, 0.264444) (5, 1.373611) (6, 4.216667) (7, 7.141111) (8, 9.821389) (9, 11.324444) (10, 11.724167)}; 
		\addplot+[mark=none]	 coordinates {(1, 0.020000) (2, 0.042222) (3, 0.057778) (4, 0.166667) (5, 0.724444) (6, 2.096389) (7, 5.029722) (8, 7.893333) (9, 9.550000) (10, 10.000000)}; 
\end{groupplot}\end{tikzpicture}
\caption{Cumulative \# of \fotwo{} sentences (a), the  expected time to fill in the database (b), and the time needed to generate sentences (c) with up to $x$ literals. At most five literals per clause, at most two clauses per sentence, one unary, and one binary predicate.}\label{fig:fo52}
\end{figure*}


\begin{figure*}

\begin{tikzpicture}\begin{groupplot}[group style={group size=3 by 1},height=5cm,width=5.4cm,xlabel=layer,cycle list/Dark2,  xticklabels={},  extra x ticks={0,1,2,3,4,5,6,7,8,9,10} ]

	\nextgroupplot[ylabel={\# sentences},ymode=log, xlabel={(a)},cycle list name=shortList,legend style = {at = {(0.05, 1.1)}, anchor = south west, legend cell  align = left, legend columns = 3}]

		\addplot+[mark=none]	 coordinates {(1, 32.000000) (2, 560.000000) (3, 5900.000000) (4, 40926.000000) (5, 200506.000000)}; \addlegendentry{baseline} 
		\addplot+[mark=none]	 coordinates {(1, 32.000000) (2, 472.000000) (3, 4636.000000) (4, 24456.000000) (5, 81326.000000) (6, 190984.000000) (7, 337364.000000)
  }; \addlegendentry{\emph{ Tautologies \& Contradictions }} 
		\addplot+[mark=none]	 coordinates {(1, 32.000000) (2, 472.000000) (3, 4636.000000) (4, 24456.000000) (5, 81326.000000) (6, 190984.000000) (7, 337364.000000) (8, 471514.000000) (9, 549106.000000)
  }; \addlegendentry{\emph{Isomorphic Sentences}} 
		\addplot+[mark=none]	 coordinates {(1, 16.000000) (2, 224.000000) (3, 1624.000000) (4, 7346.000000) (5, 22442.000000) (6, 50680.000000) (7, 87609.000000) (8, 121376.000000) (9, 140774.000000) (10, 146498.000000)}; \addlegendentry{\emph{ Negations }} 
		\addplot+[mark=none]	 coordinates {(1, 12.000000) (2, 145.000000) (3, 959.000000) (4, 4299.000000) (5, 12992.000000) (6, 29239.000000) (7, 50369.000000) (8, 69864.000000) (9, 81210.000000) (10, 84754.000000)}; \addlegendentry{\emph{ Permuting Arguments }} 
		\addplot+[mark=none]	 coordinates {(1, 9.000000) (2, 121.000000) (3, 796.000000) (4, 3601.000000) (5, 11360.000000) (6, 26726.000000) (7, 47474.000000) (8, 66888.000000) (9, 78234.000000) (10, 81778.000000)}; \addlegendentry{\emph{ Reflexive Atoms }} 
		\addplot+[mark=none]	 coordinates {(1, 9.000000) (2, 100.000000) (3, 473.000000) (4, 1705.000000) (5, 4426.000000) (6, 9045.000000) (7, 14623.000000) (8, 19945.000000) (9, 23655.000000) (10, 25042.000000)}; \addlegendentry{\emph{ Subsumption }} 
		\addplot+[mark=none]	 coordinates {(1, 7.000000) (2, 94.000000) (3, 453.000000) (4, 1652.000000) (5, 4322.000000) (6, 8902.000000) (7, 14480.000000) (8, 19802.000000) (9, 23512.000000) (10, 24899.000000)}; \addlegendentry{\emph{ Trivial Constraints }} 
		\addplot+[mark=none]	 coordinates {(1, 7.000000) (2, 91.000000) (3, 405.000000) (4, 1367.000000) (5, 3418.000000) (6, 6857.000000) (7, 10949.000000) (8, 14835.000000) (9, 17337.000000) (10, 18175.000000)}; \addlegendentry{\emph{ Cell Graph Isomorphism } } 
		\addplot+[mark=none]	 coordinates {(1, 7) (2, 71) (3, 248) (4, 667) (5, 1467) 
  }; \addlegendentry{lower bound spectra} 
	\nextgroupplot[ylabel={estimated time [h]}, xshift=1.0cm, xlabel={(b)},cycle list name=shortList]

		\addplot+[mark=none]	 coordinates {(1, 0.000010) (2, 0.662413) (3, 14.339643) (4, 138.667098) (5, 1144.499324)}; 
		\addplot+[mark=none]	 coordinates {(1, 0.000010) (2, 0.666858) (3, 11.774282) (4, 109.197354) (5, 695.956707) (6, 3019.225163) (7, 4713.280076)
  }; 
		\addplot+[mark=none]	 coordinates {(1, 0.000010) (2, 0.668802) (3, 11.789282) (4, 109.223465) (5, 696.093652) (6, 3019.666552) (7, 4714.030354) (8, 8199.755617) (9, 11485.025282)
  }; 
		\addplot+[mark=none]	 coordinates {(1, 0.000010) (2, 0.292313) (3, 6.122592) (4, 37.421033) (5, 211.469218) (6, 835.253642) (7, 1322.644186) (8, 2201.738921) (9, 3028.546915) (10, 3522.916688)}; 
		\addplot+[mark=none]	 coordinates {(1, 0.000010) (2, 0.200888) (3, 3.667599) (4, 22.106933) (5, 123.721818) (6, 482.926858) (7, 763.389982) (8, 1266.484369) (9, 1743.897112) (10, 2033.054657)}; 
		\addplot+[mark=none]	 coordinates {(1, 0.000010) (2, 0.197014) (3, 3.123318) (4, 18.327991) (5, 96.538508) (6, 390.084807) (7, 644.600273) (8, 961.060916) (9, 1462.773435) (10, 1755.608746)}; 
		\addplot+[mark=none]	 coordinates {(1, 0.000010) (2, 0.197292) (3, 2.603596) (4, 11.115792) (5, 45.843456) (6, 149.580435) (7, 227.737245) (8, 314.960634) (9, 453.726316) (10, 549.796029)}; 
		\addplot+[mark=none]	 coordinates {(1, 0.000010) (2, 0.179850) (3, 2.550197) (4, 10.717758) (5, 44.419905) (6, 146.065922) (7, 223.357339) (8, 310.248506) (9, 448.797521) (10, 544.806123)}; 
		\addplot+[mark=none]	 coordinates {(1, 0.009167) (2, 0.194562) (3, 2.570952) (4, 9.995400) (5, 37.774905) (6, 116.329855) (7, 175.145546) (8, 239.851473) (9, 341.573006) (10, 406.509071)}; 
	\nextgroupplot[ylabel={time [h]}, xshift=0.2cm, xlabel={(c)},cycle list name=shortList]

		\addplot+[mark=none]	 coordinates {(1, 0.000010) (2, 0.000288) (3, 0.003899) (4, 0.041677) (5, 0.387510)}; 
		\addplot+[mark=none]	 coordinates {(1, 0.000010) (2, 0.004732) (3, 0.057232) (4, 0.717510) (5, 3.612510) (6, 10.206954) (7, 22.538621)
  }; 
		\addplot+[mark=none]	 coordinates {(1, 0.000010) (2, 0.006677) (3, 0.072232) (4, 0.743621) (5, 3.749454) (6, 10.648343) (7, 23.288899) (8, 39.116954) (9, 47.911399)
  }; 
		\addplot+[mark=none]	 coordinates {(1, 0.000010) (2, 0.002510) (3, 0.026677) (4, 0.215843) (5, 1.026677) (6, 2.794454) (7, 6.166677) (8, 9.869732) (9, 11.953621) (10, 12.559454)}; 
		\addplot+[mark=none]	 coordinates {(1, 0.000010) (2, 0.001954) (3, 0.014454) (4, 0.123899) (5, 0.618066) (6, 1.636121) (7, 3.614732) (8, 5.828343) (9, 7.095843) (10, 7.448066)}; 
		\addplot+[mark=none]	 coordinates {(1, 0.000010) (2, 0.001399) (3, 0.011399) (4, 0.092788) (5, 0.505288) (6, 1.463621) (7, 3.445566) (8, 5.832232) (9, 7.188343) (10, 7.557510)}; 
		\addplot+[mark=none]	 coordinates {(1, 0.000010) (2, 0.001677) (3, 0.014177) (4, 0.121399) (5, 0.638621) (6, 1.768066) (7, 4.013621) (8, 6.799732) (9, 8.401677) (10, 8.838621)}; 
		\addplot+[mark=none]	 coordinates {(1, 0.000010) (2, 0.001399) (3, 0.012232) (4, 0.090288) (5, 0.498343) (6, 1.460010) (7, 3.481121) (8, 5.935010) (9, 7.320288) (10, 7.696121)}; 
		\addplot+[mark=none]	 coordinates {(1, 0.009167) (2, 0.016111) (3, 0.040833) (4, 0.174444) (5, 0.708056) (6, 1.920833) (7, 4.218056) (8, 6.981111) (9, 8.548889) (10, 8.990833)}; 
\end{groupplot}\end{tikzpicture}
\caption{Cumulative \# of \ctwo{} sentences (a), the expected time to fill in the database (b), and the time needed to generate sentences with up to $x$ literals. At most five literals per clause, at most two clauses per sentence, one unary and one binary predicate, $k \leq 1$.}\label{fig:c52}\end{figure*}


We refer to the \Cref{sec:experimentsDetails} for detailed information about the setup of the experiments and \Cref{sec:evaluationPruning} for more experiments of the pruning techniques.

\subsection{An Initial Database Construction}\label{sec:experiments:hits}

\begin{table*}[h]
\caption{A sample of sequences that are combinatorial spectra of sentences generated by our algorithm that also appear in OEIS.}\label{tab:selectedHits}
\begin{tabular}{p{6.6cm}lp{6.6cm}}
    \bf Sentence & \bf OEIS ID  & \bf OEIS name 
   \\ \hline 
    $(\forall x \exists^{=1} y  B(x, y)) \land (\forall x \exists^{=1} y  B(y, x)) \land (\forall x \forall y  B(x, x) \lor B(x, y) \lor \lnot B(y, x))$ & A85 & Number of self-inverse permutations on $n$ letters, also known as involutions; number of standard Young tableaux with $n$ cells.\\ \hline
    $(\forall x \exists^{=1} y \; B(x, y)) \land (\forall x \exists^{=1} y \; B(y, x))$ & A142 & Factorial numbers: $n! = 1\cdot 2\cdot 3\cdot 4\cdot ...\cdot n$ (order of symmetric group $S_n$, number of permutations of $n$ letters). \\ \hline
    $(\forall x  B(x, x)) \land (\forall x \exists^{=1} y \lnot B(x, y)) \land (\forall x \exists^{=1} y \lnot B(y, x))$ & A166 & Subfactorial or rencontres numbers, or derangements: number of permutations of $n$ elements with no fixed points.\\ \hline
    $(\forall x \forall y  B(x, y) \lor \lnot B(y, x)) \land (\exists x B(x, x)) \land (\forall x \exists^{=1} y  \lnot B(x, y))$ & A1189 & Number of degree-$n$ permutations of order exactly~2.\\ \hline
    $(\forall x \forall y \; U(x) \lor B(x, y)) \land (\forall x \forall y \; \lnot U(x) \lor B(y, x))$ & A47863 & Number of labeled graphs with 2-colored nodes where black nodes are only connected to white nodes and vice versa.\\ \hline
    $(\forall x \; B(x, x)) \land (\forall x \exists y \; \lnot B(x, y)) \land (\forall x \exists y \; \lnot B(y, x))$ & A86193 & Number of $n \times n$ matrices with entries in $\{0,1\}$ with no zero row, no zero column and with zero main diagonal.\\ \hline
    $(\forall x \forall y \; U(x) \lor \lnot U(y) \lor B(x, y)) \land (\forall x \exists^{=1} y \; \lnot B(x, y))$ & A290840 & $	a(n) = n! \cdot [x^n] \frac{\exp(n \cdot x)}{1 + LambertW(-x)}$. \\ \hline
\end{tabular}
\end{table*}

Apart from the experiments in which we compared the benefits of the proposed pruning methods, we also used our algorithm to generate an initial database of combinatorial sequences. For that we let the sentence generator run for five days to obtain a collection of sentences and their combinatorial spectra on a machine with 500 GB RAM, 128 processors (we used multi-threading). We used a five-minute time limit for combinatorial spectrum computation of a sequence. 

The result was a database containing over 26,000 unique integer sequences. 
For each of the sequences in our database, we queried OEIS to determine if the sequence matches a sequence which is already in OEIS. 
We found that 301 of the sequences were present in OEIS---this makes $\approx$1.2\% of the sequences we generated. This may not sound like much, but it is certainly non-negligible. Moreover, our goal was to generate primarily new sequences. We show several interesting generated sequences that happened to be in OEIS in Table \ref{tab:selectedHits}.

An example of an interesting sequence is the last one in Table \ref{tab:selectedHits}. This sequence does not have any combinatorial characterization in OEIS. We can obtain such a characterization from the \ctwo{} sentence that generated it:\footnote{For easier readability, we replaced the predicate $B$ by its negation, which does not change the spectrum.} $(\forall x \forall y \; U(x) \lor \lnot U(y) \lor \lnot B(x, y)) \land (\forall x \exists^{=1} y \; B(x, y))$. This can be interpreted as follows: {\em We are counting configurations consisting of a function $b : [n] \rightarrow [n]$ and a set $U \subseteq [n]$ that satisfy that if $y = b(x)$ and $y \in U$ then $x \in U$}. While this may not be a profound combinatorial problem, it provides a combinatorial interpretation for the sequence at hand---we would not be able to find it without the database.

Next we discuss several examples of arguably natural combinatorial sequences that were constructed by our algorithm which are not present in OEIS. The first of these examples is the sequence $0$, $0$, $6$, $72$, $980$, $15360$, $\dots$ generated by the sentence $(\forall x\; \neg B(x, x)) \land (\exists x \forall y\; \neg B(y, x)) \land (\forall x \exists^{=1} y\; B(x, y))$. We can interpret it as counting the number of functions $f : [n] \rightarrow [n]$ without fixed points and with image not equal to $[n]$. Another example is the sequence $1$, $7$, $237$, $31613$, $16224509$, $31992952773$, $\dots$, which corresponds to the sentence
$(\forall x \exists y \; B(x, y)) \land (\exists x \forall y\; B(x, y) \vee B(y, x))$ and counts directed graphs on $n$ vertices in which every vertex has non-zero out-degree and there is a vertex that is connected to all other vertices (including to itself) by either an outgoing or incoming edge. Yet another example is the sequence $1$, $5$, $127$, $12209$, $4329151$, $5723266625$, $\dots$, corresponding to the sentence $(\forall x \exists y\; B(x, y)) \land (\exists x \forall y\; B(x, y))$, which counts directed graphs where every vertex has non-zero out-degree and at least one vertex has out-degree $n$, which is also the same as the number of binary matrices with no zero rows and at least one row containing all ones. These examples correspond to the simpler structures in the database, there are others which are more complex (and also more difficult to interpret). For example, another sequence $0$, $3$, $43$, $747$, $22813$, $1352761$, $\dots$ constructed by our algorithm, given by the sentence $(\forall x\; \neg B(x, x)) \land (\forall x \forall y\; \neg B(x, y) \vee B(y, x)) \land (\exists x \forall y\; \neg B(x, y) \vee \neg U(y)) \land (\exists x \exists y\; B(x, y))$, counts 
undirected graphs without loops with at least one edge and with vertices labeled by two colors, red and black (red corresponding to $U(x)$, and black corresponding to $\neg U(x)$) such that there is at least one vertex not connected to any of the red vertices (note that this vertex can itself be red). We could keep on listing similar sequences, but we believe the handful we showed here give sufficient idea about the kind of sequences one could find in the database constructed by our system.

\section{Related Work}\label{sec:relatedWork}

To our best knowledge, there has been no prior work on automated generation of combinatorial sequences. However, there were works that intersect with the work presented in this paper in certain aspects. The most closely related are works on lifted inference \citep{DBLP:conf/ijcai/Poole03,DBLP:conf/uai/GogateD11a,Lifting/FO2-domain-liftable,Lifting/Skolemization,Lifting/FO3-intractable,Lifting/C2-domain-liftable}; this work would not be possible without lifted inference. We directly use the algorithms, even though re-implemented, as well as the concept of cell graphs from \citep{Lifting/FO2-cell-graphs}. The detection of isomorphic sentences is similar to techniques presented in \citep{DBLP:conf/aaai/BremenD0RM21}, however, that work focused on propositional logic problems, whereas here we use these techniques for problems with first-order logic sentences. There were also works on automated discovery in mathematics, e.g. \citep{colton2002hr,davies2021advancing} or the database \url{http://sequencedb.net}, but as far as we know, none in enumerative combinatorics that would be similar to ours. The closest line of works at the intersection of combinatorics and artificial intelligence are the works \citep{suster2021mapping} and \citep{totis2023lifted}. However, those works do not attempt to generate new sequences or new combinatorics results, as they mostly aim at solving textbook-style combinatorial problems, which is still a highly non-trivial problem too, though. Finally, there are several recent works that use OEIS sequences as inputs for program synthesis, e.g., \citep{pmlr-v162-d-ascoli22a,https://doi.org/10.48550/arxiv.2301.11479}. The goal of such works is orthogonal to ours and it would be interesting to see whether we could get interesting synthesized programs if we used our combinatorial sequences, which are not present in OEIS, to these systems.

\section{Conclusion}\label{sec:conclusion}
We have introduced a method for constructing a database of integer sequences with a combinatorial interpretation and used it to generate a small initial database consisting of more than 26k unique sequences, of which a non-negligible fraction appears to have been studied, which is a sign that we are able to generate interesting integer sequences automatically. Our approach has two key components: an existing lifted-inference algorithm \citep{Lifting/FO2-cell-graphs} that computes sequences from first-order logic sentences and the new method for generation of first-order sentences which successfully prunes huge numbers of redundant sentences.

\appendix

\section{Implementation Details}\label{sec:implementation}

We implemented the sentence generator described in {Section 3.2} with all of the pruning techniques from {Table 1} in Java using the following dependencies:  {Sat4J}\footnote{\url{https://www.sat4j.org/}}, {supertweety}\footnote{\url{https://github.com/supertweety/LogicStuff}}, and Prover9\footnote{\url{https://www.cs.unm.edu/~mccune/prover9/}}. A pseudocode of the algorithm generating sentences is depicted in~\Cref{alg:algorithm}; the output \emph{sentences} contains all generated \ctwo{} sentences which are further inserted into the database and their combinatorial spectra are computed.

\begin{algorithm}[tb]
    \caption{Pseudocode of the \ctwo{} sentence generator}
    \label{alg:algorithm}
    \textbf{Parameter}: literals limit as $ML$, clauses limits as $MC$, \# unary and binary predicates as $UP$ and $BP$\\
    \textbf{Output}: sentences
    \begin{algorithmic}[1]
        \STATE sentences $\leftarrow \emptyset$ 
        \STATE hidden $\leftarrow \emptyset$ 
        \STATE layer $\leftarrow \{\emptyset\}$
        \FOR{$i \in [1, 2, 3,\dots, ML \times MC ]$}
            \STATE nextLayer $\leftarrow \emptyset$
            \FOR{sentence $\in$ layer} 
                \FOR{child $\in$ refinements(sentence, $ML$, $MC$, $UP$, $BP$)}
                    \IF{is-redundant(child, sentences $\cup$ hidden)}
                        \IF{is-safe-to-delete(child, sentences $\cup$ hidden)}
                            \STATE discard the sentence \emph{child}
                        \ELSE
                            \STATE nextLayer $\leftarrow $ nextLayer $\cup \{$child$\}$  
                            \STATE hidden $\leftarrow $ hidden $\cup \{$child$\}$    
                        \ENDIF
                    \ELSE
                        \STATE nextLayer $\leftarrow $ nextLayer $\cup \{$child$\}$    
                        \STATE sentences $\leftarrow $ sentences $\cup \{$child$\}$    
                    \ENDIF
                \ENDFOR
            \ENDFOR
            \STATE layer $\leftarrow$ nextLayer
        \ENDFOR
        \STATE \textbf{return} sentences
    \end{algorithmic}
\end{algorithm}

To avoid the possibility of getting stuck while checking whether a sentence is a tautology, a contradiction, or neither of those, we use a short time limit, i.e. 30 seconds, for Prover9's execution, using it effectively as a soft filter. In general, the check is an undecidable problem, and for \ctwo{} it is {NEXPTIME}-complete~\citep{pratt2005complexity}. 

The WFOMC computation, including cell graph construction and \ctwo-related reductions, was implemented in the Julia programming language \citep{Julia/Julia}.
We made use of the Nemo.jl package \cite{Julia/Nemo} for polynomial representation and manipulation.

\subsection{OEIS Hits}

\Cref{tab:moreSelectedHits} shows all OEIS hits found during the initial database construction described in \Cref{sec:experiments:hits}.

\section{Further Details}

In this section, we discuss remaining technical issues.

\subsection{Complexity of Combinatorial Equivalence}

The problem of deciding whether two sentences have the same combinatorial spectrum is no easier than checking whether they are equivalent, which can be seen as follows. Let one of the sentences be a contradiction. Checking whether the other sentence has the same combinatorial spectrum, i.e., $0,0,0,\dots$, is equivalent to checking whether it is also a contradiction. This is only a complexity lower bound, but it already shows that checking equivalence of combinatorial spectra of \ctwo{} sentences is NEXPTIME-hard, which follows from the classical results on the complexity of satisfiability checking in \ctwo{} \cite{pratt2005complexity}. 
The exact complexity of deciding whether two \ctwo{} sentences generate the same combinatorial spectra remains an interesting open problem.

\subsection{Isomorphism of Sentences}

Here, we give a formal definition of isomorphism of two \ctwo{} sentences. We only consider sentences in the form described in \Cref{sec:sfinder}: 

\begin{align}
\bigwedge_{\begin{array}{c} Q_1 \in \{ \forall, \exists, \exists^{=1}, \dots, \exists^{=K} \},\\ Q_2 \in \{ \forall, \exists, \exists^{=1}, \dots, \exists^{=K} \} \end{array}} \bigwedge_{i=1}^{M} Q_1 x Q_2 y \; \Phi^{Q1,Q2}_i(x,y) \wedge \notag \\ 
\bigwedge_{Q \in \{ \forall, \exists, \exists^{=1}, \dots, \exists^{=K} \}} \bigwedge_{i=1}^{M'} Q x  \; \Phi^{Q}_i(x).\label{eq:sentenceFormAppendix}
\end{align}
where each $\Phi^{Q1,Q2}_i(x,y)$ and $\Phi^{Q}_i(x)$ is a disjunction of literals.

We say that the clause $Q_1 x Q_2 y \; \Phi(x,y)$ is isomorphic to the clause $Q_1' x Q_2' y \; \Psi(x,y)$ if one of the following conditions holds:
\begin{enumerate}
    \item If $Q_1 = Q_1' = Q_2 = Q_2'$ and there exists a bijection $f : \{x,y \} \rightarrow \{x,y\}$ such that the set of literals of $f(\Phi(x,y))$ is the same as the set of literals of $\Psi(x,y)$.
    \item If $Q_1 = Q_1' \neq Q_2 = Q_2'$ and the set of literals of $\Phi(x,y)$ is the same as the set of literals of $\Psi(x,y)$.
\end{enumerate}

We say that two sentences of the form (\ref{eq:sentenceForm}) are isomorphic if, for every clause from one sentence, we can find a clause from the other sentence which it is isomorphic to.

Finally, we extend this with isomorphism that allows renaming predicates: We say that two sentences of the form (\ref{eq:sentenceFormAppendix}) are isomorphic by renaming of predicates if there exists a bijection between their predicates that makes the two sentences isomorphic according to the definition of isomorphism above.

\subsection{Experiments Setup}\label{sec:experimentsDetails}

This section contains all details needed to reproduce experiments done in \Cref{sec:experiments}. Namely, the \fotwo{} experiment visualized in \Cref{fig:fo52} had the following language restrictions: at most five literals per clause, at most two clauses per sentence, at most one unary and one binary predicate. 

These restrictions also apply to the second experiment visualized in \Cref{fig:c52} which concerned \ctwo{} with $k \leq 1$, i.e. only quantifiers of type $\forall x \exists^{k=1} y$, $\exists^{k=1} x$, etc. We further forbid tuples of quantifiers where one is in the form $\exists^{=k} x$ and the other is either $\exists^{=l} y$ or $\exists y$, i.e. $\exists x \exists^{=k} y$, $\exists^{=k} x \exists y$, $\exists^{=k} x \exists^{=l} y$. Those three combinations do not scale well, so we would not be able to compute their combinatorial spectra within the five-minute limit we used for filling in the database. For the same reason, we restricted the number of literals to at most one in a clause with a counting quantifier.

Each sentence generator was executed with 51GB of memory and 48 hours of computation time. A missing part of a curve means that the corresponding generator exceeded one of the limits.

\section{Further Evaluation of Pruning Techniques}\label{sec:evaluationPruning}

We ran two more experiments that were focused purely on the presented pruning techniques rather than on the construction of a database of integer sequences.

The first experiment concerns the generation of \fotwo{} and \ctwo{} sentences with up to six literals per a clause; the rest of the setup stays the same as in \Cref{sec:experimentsDetails}.  We did not compute combinatorial spectra for those sentences, thus we only show the number of generated sentences and the runtime of sentence generators for \fotwo{} and \ctwo{} in  \Cref{fig:fo62} and \Cref{fig:c62}, respectively.
 
\begin{figure*}

\begin{tikzpicture}\begin{groupplot}[group style={group size=2 by 1},height=5cm,width=8.0cm,xlabel=layer,cycle list/Dark2,  xticklabels={},  extra x ticks={0,1,2,3,4,5,6,7,8,9,10,11,12} ]

	\nextgroupplot[ylabel={\# sentences},ymode=log, xlabel={(a)},cycle list name=shortList,legend style = {at = {(0.1, 1.1)}, anchor = south west, legend cell  align = left, legend columns = 3}]

		\addplot+[mark=none]	 coordinates {(1, 20.000000) (2, 334.000000) (3, 3890.000000) (4, 31820.000000) (5, 174188.000000)}; \addlegendentry{baseline} 
		\addplot+[mark=none]	 coordinates {(1, 20.000000) (2, 268.000000) (3, 2934.000000) (4, 19130.000000) (5, 71668.000000) (6, 178842.000000) (7, 326422.000000) (8, 469272.000000)}; \addlegendentry{\emph{ Tautologies \& Contradictions }} 
		\addplot+[mark=none]	 coordinates {(1, 20.000000) (2, 268.000000) (3, 2934.000000) (4, 19130.000000) (5, 71668.000000) (6, 178842.000000) (7, 326422.000000) (8, 469272.000000)}; \addlegendentry{\emph{Isomorphic Sentences}} 
		\addplot+[mark=none]	 coordinates {(1, 10.000000) (2, 119.000000) (3, 1002.000000) (4, 5688.000000) (5, 19634.000000) (6, 47254.000000) (7, 84503.000000) (8, 120516.000000) (9, 144706.000000) (10, 155958.000000) (11, 159138.000000) (12, 159624.000000)}; \addlegendentry{\emph{ Negations }} 
		\addplot+[mark=none]	 coordinates {(1, 8.000000) (2, 84.000000) (3, 616.000000) (4, 3386.000000) (5, 11437.000000) (6, 27329.000000) (7, 48687.000000) (8, 69619.000000) (9, 83949.000000) (10, 90929.000000) (11, 93027.000000) (12, 93389.000000)}; \addlegendentry{\emph{ Permuting Arguments }} 
		\addplot+[mark=none]	 coordinates {(1, 6.000000) (2, 65.000000) (3, 499.000000) (4, 2785.000000) (5, 9924.000000) (6, 24935.000000) (7, 45911.000000) (8, 66762.000000) (9, 81092.000000) (10, 88072.000000) (11, 90170.000000) (12, 90532.000000)}; \addlegendentry{\emph{ Reflexive Atoms }} 
		\addplot+[mark=none]	 coordinates {(1, 6.000000) (2, 44.000000) (3, 269.000000) (4, 1240.000000) (5, 3612.000000) (6, 8008.000000) (7, 13711.000000) (8, 19625.000000) (9, 24128.000000) (10, 26715.000000) (11, 27577.000000) (12, 27750.000000)}; \addlegendentry{\emph{ Subsumption }} 
		\addplot+[mark=none]	 coordinates {(1, 4.000000) (2, 42.000000) (3, 253.000000) (4, 1191.000000) (5, 3512.000000) (6, 7869.000000) (7, 13555.000000) (8, 19469.000000) (9, 23972.000000) (10, 26559.000000) (11, 27421.000000) (12, 27594.000000)}; \addlegendentry{\emph{ Trivial Constraints }} 
		\addplot+[mark=none]	 coordinates {(1, 4.000000) (2, 40.000000) (3, 216.000000) (4, 923.000000) (5, 2642.000000) (6, 5869.000000) (7, 10027.000000) (8, 14218.000000) (9, 17171.000000) (10, 18679.000000) (11, 19104.000000) (12, 19177.000000)}; \addlegendentry{\emph{ Cell Graph Isomorphism } } 

 \nextgroupplot[ylabel={time [h]}, xshift=0.8cm, xlabel={(b)},cycle list name=shortList]

		\addplot+[mark=none]	 coordinates {(1, 0.000010) (2, 0.000288) (3, 0.004454) (4, 0.036121) (5, 0.301399)}; 
		\addplot+[mark=none]	 coordinates {(1, 0.002778) (2, 0.068611) (3, 0.814722) (4, 1.595278) (5, 4.988611) (6, 12.797222) (7, 26.589167) (8, 43.439167)}; 
		\addplot+[mark=none]	 coordinates {(1, 0.002778) (2, 0.076667) (3, 0.860556) (4, 1.838056) (5, 6.049167) (6, 14.286111) (7, 28.129167) (8, 44.997222)}; 
		\addplot+[mark=none]	 coordinates {(1, 0.002778) (2, 0.034444) (3, 0.354444) (4, 1.094167) (5, 2.075556) (6, 4.311111) (7, 8.477500) (8, 12.967222) (9, 15.731389) (10, 17.085000) (11, 17.479722) (12, 17.534722)}; 
		\addplot+[mark=none]	 coordinates {(1, 0.002500) (2, 0.028333) (3, 0.228611) (4, 0.993611) (5, 1.575278) (6, 2.853333) (7, 5.285000) (8, 8.093611) (9, 9.987500) (10, 10.872500) (11, 11.154167) (12, 11.197222)}; 
		\addplot+[mark=none]	 coordinates {(1, 0.001667) (2, 0.018611) (3, 0.162500) (4, 0.930833) (5, 1.444167) (6, 2.598333) (7, 5.085556) (8, 8.075278) (9, 10.037222) (10, 10.908889) (11, 11.185278) (12, 11.226389)}; 
		\addplot+[mark=none]	 coordinates {(1, 0.002222) (2, 0.018889) (3, 0.161111) (4, 0.928333) (5, 1.455000) (6, 2.630833) (7, 5.147222) (8, 8.183333) (9, 10.184444) (10, 11.082778) (11, 11.357222) (12, 11.396111)}; 
		\addplot+[mark=none]	 coordinates {(1, 0.002222) (2, 0.019167) (3, 0.169167) (4, 0.941389) (5, 1.523611) (6, 2.804167) (7, 5.532222) (8, 8.700556) (9, 10.740556) (10, 11.626667) (11, 11.885833) (12, 11.925833)}; 
		\addplot+[mark=none]	 coordinates {(1, 0.020278) (2, 0.053611) (3, 0.208889) (4, 0.971944) (5, 1.629722) (6, 3.093611) (7, 6.118056) (8, 9.546667) (9, 11.862500) (10, 12.823056) (11, 13.096667) (12, 13.138611)}; 
\end{groupplot}\end{tikzpicture}
\caption{Cumulative \# of \fotwo{} sentences (a) and the time needed to generate sentences (b) with up to $x$ literals. At most six literals per clause, at most two clauses per sentence, one unary, and one binary predicate.}\label{fig:fo62}\end{figure*}
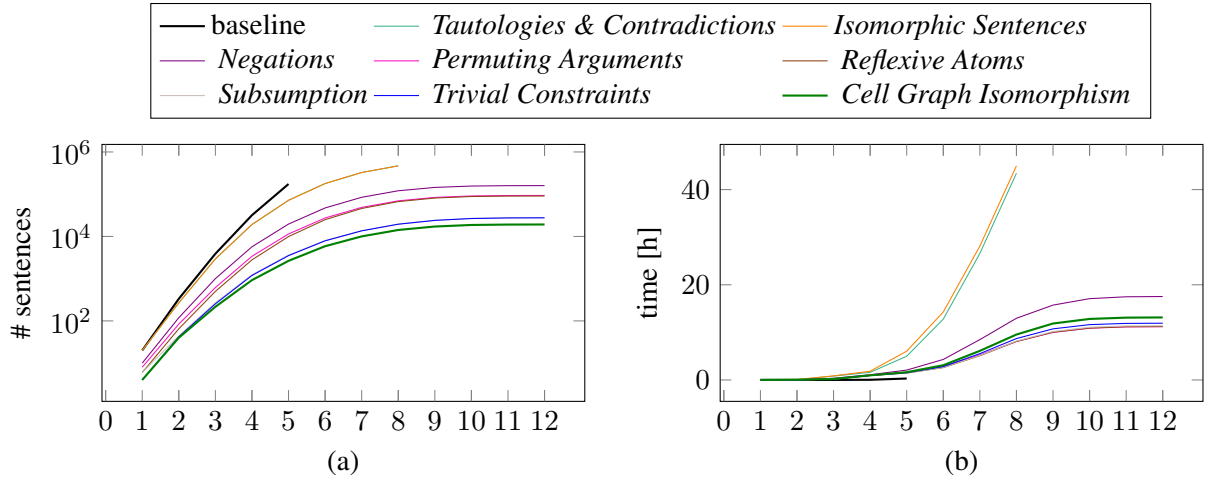

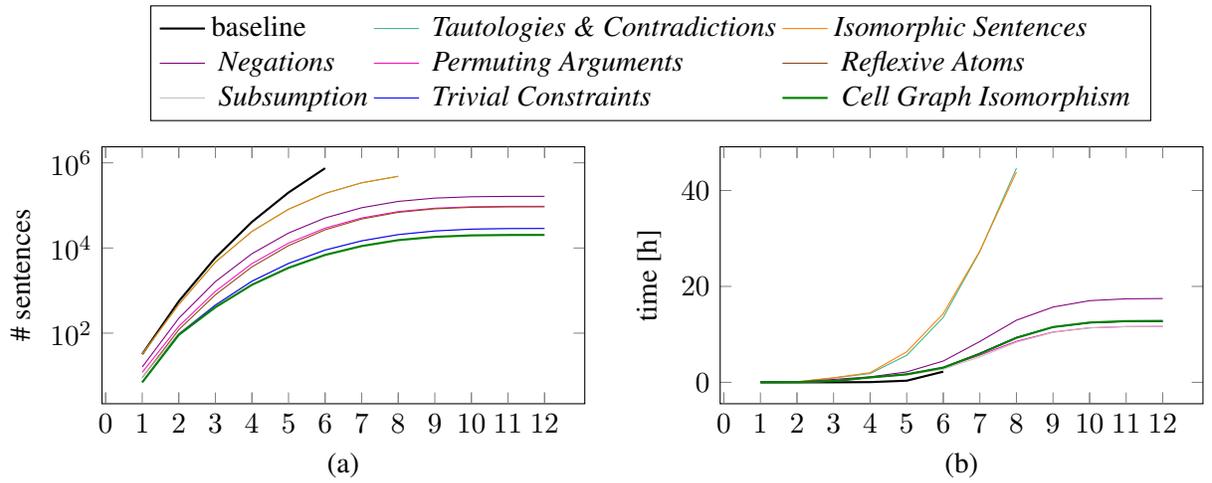
\begin{figure*}

\begin{tikzpicture}\begin{groupplot}[group style={group size=2 by 1},height=5cm,width=8.0cm,xlabel=layer,cycle list/Dark2,  xticklabels={},  extra x ticks={0,1,2,3,4,5,6,7,8,9,10,11,12} ]

	\nextgroupplot[ylabel={\# sentences},ymode=log, xlabel={(a)},cycle list name=shortList,legend style = {at = {(0.1, 1.1)}, anchor = south west, legend cell  align = left, legend columns = 3}]

		\addplot+[mark=none]	 coordinates {(1, 32.000000) (2, 560.000000) (3, 5900.000000) (4, 40926.000000) (5, 200506.000000) (6, 747584.000000)}; \addlegendentry{baseline} 
		\addplot+[mark=none]	 coordinates {(1, 32.000000) (2, 472.000000) (3, 4636.000000) (4, 24456.000000) (5, 81326.000000) (6, 191044.000000) (7, 339344.000000) (8, 482194.000000)}; \addlegendentry{\emph{ Tautologies \& Contradictions }} 
		\addplot+[mark=none]	 coordinates {(1, 32.000000) (2, 472.000000) (3, 4636.000000) (4, 24456.000000) (5, 81326.000000) (6, 191044.000000) (7, 339344.000000) (8, 482194.000000)}; \addlegendentry{\emph{Isomorphic Sentences }} 
		\addplot+[mark=none]	 coordinates {(1, 16.000000) (2, 224.000000) (3, 1624.000000) (4, 7346.000000) (5, 22442.000000) (6, 50698.000000) (7, 88139.000000) (8, 124152.000000) (9, 148342.000000) (10, 159594.000000) (11, 162774.000000) (12, 163260.000000)}; \addlegendentry{\emph{ Negations }} 
		\addplot+[mark=none]	 coordinates {(1, 12.000000) (2, 145.000000) (3, 959.000000) (4, 4299.000000) (5, 12992.000000) (6, 29254.000000) (7, 50730.000000) (8, 71662.000000) (9, 85992.000000) (10, 92972.000000) (11, 95070.000000) (12, 95432.000000)}; \addlegendentry{\emph{ Permuting Arguments }} 
		\addplot+[mark=none]	 coordinates {(1, 9.000000) (2, 121.000000) (3, 796.000000) (4, 3601.000000) (5, 11360.000000) (6, 26741.000000) (7, 47835.000000) (8, 68686.000000) (9, 83016.000000) (10, 89996.000000) (11, 92094.000000) (12, 92456.000000)}; \addlegendentry{\emph{ Reflexive Atoms }} 
		\addplot+[mark=none]	 coordinates {(1, 9.000000) (2, 100.000000) (3, 473.000000) (4, 1705.000000) (5, 4426.000000) (6, 9056.000000) (7, 14845.000000) (8, 20759.000000) (9, 25262.000000) (10, 27849.000000) (11, 28711.000000) (12, 28884.000000)}; \addlegendentry{\emph{ Subsumption }} 
		\addplot+[mark=none]	 coordinates {(1, 7.000000) (2, 94.000000) (3, 453.000000) (4, 1652.000000) (5, 4322.000000) (6, 8913.000000) (7, 14685.000000) (8, 20599.000000) (9, 25102.000000) (10, 27689.000000) (11, 28551.000000) (12, 28724.000000)}; \addlegendentry{\emph{ Trivial Constraints }} 
		\addplot+[mark=none]	 coordinates {(1, 7.000000) (2, 91.000000) (3, 405.000000) (4, 1367.000000) (5, 3418.000000) (6, 6865.000000) (7, 11100.000000) (8, 15291.000000) (9, 18244.000000) (10, 19752.000000) (11, 20177.000000) (12, 20250.000000)}; \addlegendentry{\emph{ Cell Graph Isomorphism } } 

	\nextgroupplot[ylabel={time [h]}, xshift=0.8cm, xlabel={(b)},cycle list name=shortList]

		\addplot+[mark=none]	 coordinates {(1, 0.000010) (2, 0.000288) (3, 0.004454) (4, 0.042510) (5, 0.336954) (6, 2.224732)}; 
		\addplot+[mark=none]	 coordinates {(1, 0.000010) (2, 0.086121) (3, 0.912788) (4, 1.826677) (5, 5.630843) (6, 13.560566) (7, 27.286954) (8, 44.648621)}; 
		\addplot+[mark=none]	 coordinates {(1, 0.001944) (2, 0.121111) (3, 0.949722) (4, 2.013333) (5, 6.370556) (6, 14.331389) (7, 27.398889) (8, 43.841111)}; 
		\addplot+[mark=none]	 coordinates {(1, 0.001944) (2, 0.059444) (3, 0.602222) (4, 1.142222) (5, 2.168333) (6, 4.438611) (7, 8.518611) (8, 12.950833) (9, 15.694444) (10, 17.025000) (11, 17.414722) (12, 17.470556)}; 
		\addplot+[mark=none]	 coordinates {(1, 0.001667) (2, 0.041389) (3, 0.345278) (4, 1.037500) (5, 1.687778) (6, 3.033333) (7, 5.667778) (8, 8.555556) (9, 10.481389) (10, 11.361667) (11, 11.618611) (12, 11.658056)}; 
		\addplot+[mark=none]	 coordinates {(1, 0.001389) (2, 0.028889) (3, 0.251944) (4, 1.009167) (5, 1.666389) (6, 3.071389) (7, 6.040833) (8, 9.346667) (9, 11.522500) (10, 12.414167) (11, 12.687500) (12, 12.728056)}; 
		\addplot+[mark=none]	 coordinates {(1, 0.001389) (2, 0.026944) (3, 0.229444) (4, 0.971667) (5, 1.519444) (6, 2.745556) (7, 5.296389) (8, 8.354722) (9, 10.413611) (10, 11.350278) (11, 11.617222) (12, 11.656944)}; 
		\addplot+[mark=none]	 coordinates {(1, 0.001667) (2, 0.028889) (3, 0.259444) (4, 1.009722) (5, 1.666944) (6, 3.073056) (7, 6.042222) (8, 9.347222) (9, 11.520833) (10, 12.412222) (11, 12.685833) (12, 12.726667)}; 
		\addplot+[mark=none]	 coordinates {(1, 0.013611) (2, 0.041667) (3, 0.260278) (4, 1.023056) (5, 1.668889) (6, 3.073889) (7, 5.944444) (8, 9.275278) (9, 11.523889) (10, 12.469722) (11, 12.760833) (12, 12.804722)}; 
\end{groupplot}\end{tikzpicture}
\caption{Cumulative \# of \ctwo{} sentences (a) and the time needed to generate sentences (b) with up to $x$ literals. At most six literals per clause, at most two clauses per sentence, one unary, and one binary predicate, $k \leq 1$.}\label{fig:c62}\end{figure*}


Finally, we ran a setup with two unary and two binary predicates, with at most three clauses, and at most two literals per clause; results for \fotwo{} and \ctwo{} are shown in \Cref{fig:fo23} and \Cref{fig:c23}, respectively. This last setup was the hardest due to the higher branching factor w.r.t. the previous ones.

\begin{figure*}
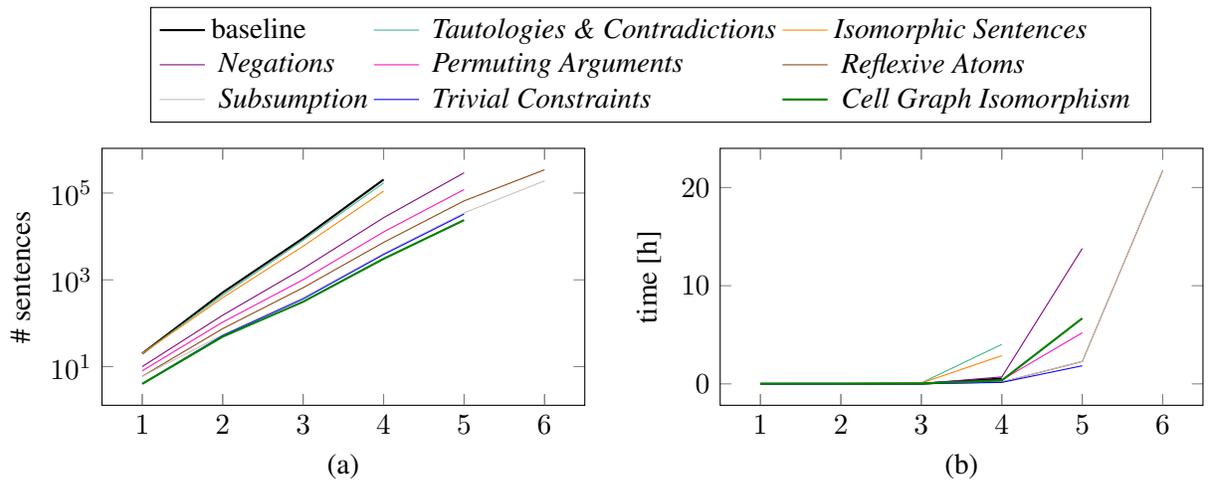



\caption{Cumulative \# of \fotwo{} sentences (a) and the time needed to generate sentences (b) with up to $x$ literals. At most two literals per clause, at most three clauses per sentence, two unary, and two binary predicate.}\label{fig:fo23}\end{figure*}

\begin{figure*}

%
\caption{Cumulative \# of \ctwo{} sentences (a) and the time needed to generate sentences (b) with up to $x$ literals. At most two literals per clause, at most three clauses per sentence, two unary, and two binary predicate, $k \leq 1$.}\label{fig:c23}\end{figure*}


The experiments clearly confirm the usefulness of the proposed pruning techniques.

\bibliographystyle{bib/named}
\bibliography{main}

\begin{longtable}[t]{p{6.5cm}lp{6.5cm}}
\caption{A bigger sample of OEIS hits.}\label{tab:moreSelectedHits}

\\ \hline \bf Sentence & \bf OEIS ID & \bf OEIS name \\ \hline

\endfirsthead

\multicolumn{3}{c}%
{{\bfseries \tablename\ \thetable{} -- continued from previous page}} \\
\hline \bf Sentence & \bf OEIS ID & \bf OEIS name \\ \hline
\endhead

\endfoot

\hline
\endlastfoot

$(\exists^{=1} x \; U(x))$ & A27 & The positive integers. Also called the natural numbers, the whole numbers or the counting numbers, but these terms are ambiguous. \\ \hline 
$(\forall x \; U(x) \vee \lnot U(x)) \land (\forall x \; \lnot U(x) \vee U(x))$ & A79 & Powers of 2: $a(n) = 2^n$. \\ \hline 
$(\forall x \exists^{=1} y \; B(x, y)) \land (\forall x \exists^{=1} y \; B(y, x)) \land (\forall x \forall y \; B(x, y) \vee \lnot B(y, x))$ & A85 & Number of self-inverse permutations on n letters, also known as involutions; number of standard Young tableaux with $n$ cells. \\ \hline 
$(\forall x \exists y \; \lnot B(y, x)) \land (\exists x \exists y \; \lnot B(x, y)) \land (\forall x \exists^{=1} y \; \lnot B(x, y))$ & A142 & Factorial numbers: $n! = 1*2*3*4*...*n$ (order of symmetric group $S_n$, number of permutations of $n$ letters). \\ \hline 
$(\forall x \exists^{=1} y \; B(x, y)) \land (\forall x \exists^{=1} y \; B(y, x)) \land (\forall x \forall y \; B(x, y) \vee B1(x, y))$ & A165 & Double factorial of even numbers: $(2n)!! = 2^n*n!$. \\ \hline 
$(\forall x \; \lnot U(x) \vee \lnot B(x, x)) \land (\forall x \exists^{=1} y \; \lnot B(x, y))$ & A169 & Number of labeled rooted trees with $n$ nodes: $n^{n-1}$. \\
\end{longtable}

\clearpage

\end{document}